\documentclass[stsy,sglanonrev]{INFORMS-STSY-Template-26-02-2026/informs4}
\usepackage{INFORMS-STSY-Template-26-02-2026/eqndefns-left} % For checking the display equation width and equation environment definitions %
\RequirePackage{tgtermes}
\RequirePackage{newtxtext}
\RequirePackage{newtxmath}
\RequirePackage{bm}
\RequirePackage{endnotes}

\OneAndAHalfSpacedXII % Current default line spacing
\usepackage{caption}
\usepackage[labelfont=sf]{subcaption}
\graphicspath{{INFORMS-STSY-Template-26-02-2026/}}

\usepackage{natbib}
 \bibpunct[, ]{(}{)}{,}{a}{}{,}%
 \def\bibfont{\small}%
\EquationsNumberedThrough    % Default: (1), (2), ...
\TheoremsNumberedThrough     % Preferred (Theorem 1, Lemma 1, Theorem 2)
\ECRepeatTheorems

\providecommand{\texorpdfstring}[2]{#1}

\newcommand{\E}[1]{\mathrm{\bf E}\left[#1\right]}
\def\P#1{\mathrm{\mathbb{P}}\left[#1\right]}

\usepackage{xparse}
\usepackage{xspace}

\newcommand{\threshsplit}[1]{\mbox{Thresh-SPLIT}(\ensuremath{#1})\xspace}

\newcommand{\SPLIT}{SPLIT\xspace}

\newcommand{\DynoSPLIT}{\texorpdfstring{\text{DynaSPLIT}}{DynaSPLIT}\xspace}
\newcommand{\Fbar}{\overline F}
\newcommand{\one}{\mathbf 1}
\newcommand{\hatA}{\hat{A}}
\makeatletter
\newcommand{\WA}{\@ifnextchar\bgroup{\WA@split}{\mathrm{WA}}}
\newcommand{\WA@split}[2]{W_{A_{#1}}(#2)}
\makeatother

\let\informsproof\proof
\let\endinformsproof\endproof
\renewenvironment{proof}[1][Proof.]{\informsproof{#1}}{\Halmos\endinformsproof}

\def\AppendixFontSize{\normalsize}

\MANUSCRIPTNO{}

\begin{document}
%%%%%%%%%%%%%%%%

% Author's names for the running heads (hidden under dblanonrev):
\RUNAUTHOR{Li, Harchol-Balter, and Scheller-Wolf}

% Title or shortened title suitable for running heads.
\RUNTITLE{\DynoSPLIT{}: Strong tail-optimal scheduling for the M/G/n}

% Full title.
\TITLE{\DynoSPLIT{}: Strong tail-optimal scheduling for the M/G/n with heavy-tailed job sizes}

% Block of authors and their affiliations (hidden under dblanonrev).
\ARTICLEAUTHORS{%
\AUTHOR{Zhouzi Li, Mor Harchol-Balter, Alan Scheller-Wolf}
\AFF{Carnegie Mellon University, \EMAIL{zhouzil@andrew.cmu.edu}, \EMAIL{harchol@cs.cmu.edu}, \EMAIL{awolf@andrew.cmu.edu}}
} % end of the block

% At most 200 words, one paragraph, no citations or equations.
\ABSTRACT{% Abstract text only; main.tex wraps it in \ABSTRACT{}.
% STSY: at most 200 words, one paragraph, no citations or equations.
% \draftplaceholder
We study the asymptotic response time tail in the M/G/$n$ queue with heavy-tailed (regularly varying) job sizes. For this setting, strongly tail optimal policies are known only at lower load, where $n-1$ servers suffice for stability. At higher load, the tail optimal scheduling has been open. This paper presents \DynoSPLIT, the first strongly tail optimal policy at high load. Together with prior work, this result shows that strong tail optimality is achievable across the entire stability region of the M/G/$n$ with heavy-tailed job sizes.
 }

% \FUNDING{}

% \KEYWORDS{}

\maketitle
%%%%%%%%%%%%%%%%%%%%%%%%%%%%%%%%%%%%%%%%%%%%%%%%%%%%%%%%%%%%%%%%%%%%%%

% !TEX root = ../main.tex
% New introduction following the whiteboard outline of 2026-09-21.
% The previous version is archived in 1-introduction-old.tex (not compiled).
% The whole section is typeset in blue to mark it as new text; delete the
% \begingroup\color{blue} ... \endgroup pair to restore black.

\section{Introduction}
\label{sec:introduction}

In many computing systems, the tail of the response time distribution matters at least as much as its mean. Service-level objectives are stated in terms of the 99th or 99.9th percentile of response time, and what users notice is the rare job that waits far longer than it should.  
  
The scheduling policy is an important tool for controlling this tail. Following the theory literature on tail-optimal scheduling \citep{wierman2012tail,boxma2007tails,nuyens2008preventing,yu2024strongly,split}, we measure the tail asymptotically: We study the decay of $\P{T>t}$ as $t\to\infty$, where $T$ denotes the response time of a job. As argued in \citet{yu2024strongly}, this asymptotic regime governs the high quantiles of response time, and it is the regime in which optimality questions are tractable. A policy is \emph{strongly tail optimal} if its $\P{T>t}$ is asymptotically no larger than under any other policy, and \emph{weakly tail optimal} if it is asymptotically within a constant factor of every other policy. The formal definitions are given in Section~\ref{sec:setting}. 

Throughout the paper we focus on the case when job sizes are heavy-tailed (or more precisely regularly varying, see Section~\ref{sec:setting}). This is representative of real world workloads \citep{crovella1997self,harchol1997exploiting,tirmazi2020borg}, where a small fraction of the jobs are enormous and account for a large fraction of the work. We defer all discussion of prior work on the light-tailed world to Section~\ref{sec:related-work}.

Under heavy-tailed workloads, for a single server, the tail-optimal policy has long been known. Let $S$ denote the job size, $\lambda$ denote the arrival rate, and $\rho:= \frac{\lambda \E{S}}{n}$ denote the load. The optimal response time tail is asymptotically $\P{T>t} \sim \P{S>(1-\rho)t}$ \citep{wierman2012tail,yu2024strongly}. This tail is attained by Shortest Remaining Processing Time (SRPT) and by many other policies that protect small jobs from large ones, including Processor Sharing, Foreground-Background, SMART, and Gittins \citep{nuyens2008preventing,borst2006sojourn,guillemin2004tail,scully2025does}. The lesson of the single-server case is that helping a large job never helps the tail: Any capacity given to a large job is taken directly from the small jobs behind it, and it is the small jobs that must be protected.

For multi-server systems, by contrast, tail analysis is still at a very early stage.  What has been shown is that the single-server lesson is wrong: SRPT is not optimal. In particular, a very recent paper by Li et al.~\citeyearpar{split} has studied the case of an M/G/$n$ with heavy-tailed job sizes that are known to the scheduler. They show that a multi-server system needs to provide a little ``sympathy'' for large jobs. Their policy \SPLIT (SymPathy for Large jobs Improves Tail latency) splits the $n$ servers into two groups: $n-1$ servers run SRPT, where the one remaining server is dedicated to running the job with the largest original size in the system.

Li et al.~\citeyearpar{split} prove that the \SPLIT policy is strongly tail optimal. This optimality, however, rests on one condition: The $n-1$ SRPT servers must be able to carry the whole load by themselves, i.e., $\rho<\frac{n-1}{n}$. Equivalently, the system must have a spare server. Throughout this paper, we refer to $\rho<\frac{n-1}{n}$ as the \emph{lower-load regime} and to $\frac{n-1}{n}\le\rho<1$ as the \emph{high-load regime}.
While \SPLIT is shown to be strongly tail optimal in the lower-load regime, it performs badly in the high-load regime.

For the high-load regime, Li et al.~\citeyearpar{split} propose a threshold policy, \threshsplit{d}, which routes jobs larger than a threshold $d$ to one server and the remaining jobs to the other $n-1$ servers. \threshsplit{d} can be made arbitrarily close to strongly tail optimal by tuning $d$; however, \threshsplit{d} is {\em not }strongly tail optimal for any $d$.

Li et al.~\citeyearpar{split} pose as an open problem the question of whether strong tail optimality can be achieved at all in the high-load regime. The goal of this paper is to close this open problem by introducing a new policy, \DynoSPLIT, that is strongly tail optimal. 
Specifically, we prove that \DynoSPLIT is strongly tail optimal in the high-load regime; recall that SPLIT~\citep{split} already achieves strong tail optimality under lower load.
Table~\ref{tab:intro:results} lists the optimal asymptotic response time tail and summarizes the known results together with the contribution of this paper.

\begin{table}[t]
\TABLE
{Tail optimality results for the M/G/$n$ with heavy-tailed job sizes when job sizes are known to the scheduler.\label{tab:intro:results}}
{\begin{tabular}{@{}lll@{}}
\hline\up\down
 & \textbf{Lower load} $\rho<\frac{n-1}{n}$ & \textbf{High load} $\frac{n-1}{n}\leq\rho<1$ \\
\hline\up
\textbf{Optimal asymptotic response time tail}  & $\P{S>t}$ & $\P{S>n(1-\rho)t}$ \\
\textbf{Arbitrarily close to strongly optimal} & \threshsplit{d} \citep{split} & \threshsplit{d} \citep{split} \\
\down\textbf{Strongly optimal} & \SPLIT \citep{split} & \DynoSPLIT [\textbf{This paper}] \\
\hline
\end{tabular}}
{}
\end{table}

In the remainder of the introduction, Section~\ref{sec:intro:highload} gives background on the two load regimes and explains why the high-load regime is fundamentally harder. Section~\ref{sec:intro:dilemma} presents the challenge behind achieving strong tail optimality at high load. Section~\ref{sec:intro:results} presents our results and the key ideas behind them.

\subsection{Background: the high-load regime versus the lower-load regime}
\label{sec:intro:highload}

Throughout, we write $\delta:=n(1-\rho)$ for the \emph{spare capacity} of the system, namely the total capacity $n$ of the servers minus the mean amount of work $n\rho$ arriving per unit time. The boundary $\rho=\frac{n-1}{n}$ between the two regimes is exactly $\delta=1$: In the lower-load regime the system has more than one server of capacity to spare, and in the high-load regime it has at most one.

The two regimes differ in two respects:

\paragraph{A spare server or not.}
The main difference is whether one large job can be given a whole server without harming the rest of the system. In the lower-load regime, a whole server can be spared: If one server is occupied by a huge job for a long time, the other $n-1$ servers still have positive spare capacity $(n-1)-n\rho=\delta-1>0$ and keep up with the arrivals.
In the high-load regime, in contrast, one huge job occupying one server already leaves the other $n-1$ servers with an offered load of $n\rho\ge n-1$, so they {have no spare capacity left}. {Roughly speaking, these $n-1$ servers then cannot keep up with the arriving work for as long as the huge job is present, and jobs arriving during this period can suffer very large response times.}

{We call such an event a \emph{catastrophe}: Huge jobs hold servers for so long that the other servers cannot keep up with the arriving work, and many jobs are delayed by large amounts.} {The} boundary $\rho=\frac{n-1}{n}$ changes {the number of huge jobs needed to cause a catastrophe}: At high load, a single huge job {is enough}, while at lower load, a catastrophe requires multiple huge jobs to be present in the system. The same distinction runs through the earlier work on multi-server FCFS queues with heavy-tailed job sizes \citep{schellerwolf1997delay,schellerwolf2000further,whitt2000impact,schellerwolf2003necessary,foss2006heavy,schellerwolf2006structural,schellerwolf2011sink,foss2012large,blanchet2016tail}, where the number of huge jobs needed to overload the remaining servers governs both the tail and the moments of the waiting time.  

\paragraph{Different lower bounds.}
The two load regimes also differ in the lower bound that a strongly tail optimal policy must meet \citep{split}. In the lower-load regime, the lower bound on the response time tail is $\P{S>t}$: A job cannot finish before its own size, so no policy can do better. \SPLIT achieves this bound. In the high-load regime, the lower bound on the response time tail is $\P{S>n(1-\rho)t}=\P{S>\delta t}$; note that this bound is higher than $\P{S>t}$ if $\delta<1$. This bound comes from comparing the $n$-server system with a single \emph{super server} of speed $n$: The super server can imitate any $n$-server policy and therefore has a better response time tail than any $n$-server policy. By the single-server result \citep{wierman2012tail,yu2024strongly}, no super-server policy has a response time tail lighter than $\P{S>\delta t}$ (see Section~\ref{sec:theory:lowerbound}).

The super-server bound is hard to attain: It requires the $n$-server system to behave, in its tail, exactly like a single server that is $n$ times as fast. So far, the performance of a super server has been matched only in the heavy-traffic limit $\rho\to1$, for the mean response time \citep{grosof2019srpt} and for the response time tail under light-tailed job sizes \citep{yu2025tale}. This paper presents the first policy to meet this super-server lower bound in the high-load regime, where the load $\rho$ is a fixed constant less than $1$.

\subsection{The challenge of attaining optimality at high load}
\label{sec:intro:dilemma}

Consider a huge job; how much capacity should be given to it while it is present? The super-server lower bound leaves no freedom in this choice: The answer must be $\delta$. We now explain why. Suppose that a policy holds back a fixed reserve $\varepsilon>0$ of capacity from every huge job, so that a huge job is served at an average rate of at most $\delta-\varepsilon$ over its lifetime. A huge job of size $x$ then stays in the system for at least $x/(\delta-\varepsilon)$ time. Turning this around, every job larger than $(\delta-\varepsilon)t$ has response time larger than $t$, so under this policy $\P{T>t}\ge\P{S>(\delta-\varepsilon)t}$. Compare this bound with the super-server lower bound $\P{S>\delta t}$. Although the two only differ in the threshold on the job size, $(\delta-\varepsilon)t$ instead of $\delta t$, that is already enough to ruin the optimality. For a regularly varying job size distribution with tail index $\alpha$, the tail behaves like a power law at large sizes, so lowering the threshold by the factor $(\delta-\varepsilon)/\delta$ raises the probability, for large $t$, by the factor $(\delta/(\delta-\varepsilon))^{\alpha}>1$. The response time tail of the policy is therefore at least a constant factor above the super-server lower bound, for all large $t$, however small the reserve $\varepsilon$ is.

Hence a strongly tail optimal policy cannot hold back any fixed reserve $\varepsilon>0$ from a huge job: The average rate at which it serves a huge job must approach $\delta$ as the size of the huge job grows. In other words, a huge job must be given essentially the entire spare capacity of the system, for as long as it is present.

However, there is a conundrum: Suppose that a huge job is given the entire spare capacity $\delta$. The huge job then remains in the system for a time proportional to its size, which is arbitrarily long. During this whole time, the remaining jobs are left with capacity $n-\delta=n\rho$, which is exactly the mean amount of work arriving per unit time. {Roughly speaking, the remaining jobs then cannot keep up with the arriving work, and jobs arriving while the huge job is present can be delayed by amounts that grow with the size of the huge job.}

This is the dilemma of the high-load regime: Strong tail optimality requires giving a huge job essentially all of the spare capacity of the system, while the remaining jobs need a strictly positive share of that same spare capacity to keep up with the arriving work. No fixed division of the spare capacity satisfies both requirements.

% The threshold policy \threshsplit{d} of \citet{split} sits on the second horn of the dilemma. Its threshold $d$ determines how much of the load is routed to the $n-1$ small-job servers and hence the reserve that these servers keep; the large-job server carries the rest of the load, so a huge job on the large-job server is served at an average rate of $\delta$ minus that reserve. The tail of \threshsplit{d} therefore approaches the best possible only as $d$ approaches the value at which the reserve vanishes, which is why \threshsplit{d} can be made arbitrarily close to strongly tail optimal but never is. The dilemma has no counterpart in the lower-load regime, where the remaining servers keep a reserve of $\delta-1>0$ for free even when the huge job holds a whole server, and it is the reason why the high-load regime needs a new idea.

\subsection{Results and key ideas}
\label{sec:intro:results}

This paper introduces \DynoSPLIT, the first policy that is strongly tail optimal in the high-load regime. Our main result (Theorem~\ref{thm:frac:optimality}) is that \DynoSPLIT attains the super-server lower bound at every load $\frac{n-1}{n}\le\rho<1$. Together with \SPLIT \citep{split}, strongly tail optimal policies now exist over the entire stability region of the M/G/$n$ with heavy-tailed, known job sizes. 

\paragraph{Key idea: divide the spare capacity dynamically.}
The dilemma of Section~\ref{sec:intro:dilemma} rules out every fixed division of the spare capacity between the largest job and the other jobs. \DynoSPLIT therefore makes the division depend on the size of the largest job. The job with the largest original size in the system, which we call the \emph{anchor}, is served at rate $\delta$ minus a \emph{margin} that depends on the anchor's size: The margin is strictly positive for every finite anchor size but decreases to zero as the anchor's size grows. The other (smaller) jobs are served in SRPT order with the remaining capacity.
A larger anchor thus receives a larger share of the spare capacity, and the share approaches the entire spare capacity as the anchor's size grows. This \emph{dynamic} division motivates the name \DynoSPLIT. Section~\ref{sec:frac:roadmap} defines the policy formally.

\paragraph{Why it works.}
A {margin} that depends on the size of the anchor works for both the anchor and the other (smaller) jobs: The {margin} is strictly positive for every finite anchor size, which is what the smaller jobs need, yet it vanishes as the anchor's size grows, which is what the anchor needs.

Under \DynoSPLIT, the {margin} shrinks exactly when the anchor is huge, and one may worry that these are the moments when the small jobs need the {margin} most. Roughly speaking, the opposite is true. The {margin} that the small jobs need is proportional not to their workload but to its fluctuations. The longer the period over which the workload is averaged, the smaller the {margin} per unit time that suffices to absorb the small-job fluctuations. A huge anchor is present for a long time, and, over a long time, a small {margin} per unit time suffices.

\paragraph{Analysis.}
Under \SPLIT in the lower-load regime, the small jobs always see a constant {margin}, which Li et al.~\citeyearpar{split} exploit to show that there is only a negligible probability of a small job's response time exceeding a constant multiple of its size. Under \DynoSPLIT in the high-load regime, the {margin} seen by a small job depends on the current anchor and can be arbitrarily small. The analysis therefore requires new techniques. Section~\ref{sec:frac:three-regimes} gives the roadmap of the proof, and Section~\ref{sec:frac:proof} contains the details.

% !TEX root = ../main.tex
% Rewritten 2026-09-22 so as not to duplicate the problem-setting section
% of the SPLIT paper.  Only the super-server lower bound is derived.  The
% previous version is archived in 3-problem-setting-old.tex (not compiled).
% The whole section is typeset in blue to mark it as new text; delete the

\section{Problem Setting and Notation}
\label{sec:setting}

This section fixes the model and the objective (Section~\ref{sec:setting:model}) and derives the super-server lower bound (Section~\ref{sec:theory:lowerbound}), which is the benchmark that Theorem~\ref{thm:frac:optimality} matches. 

\subsection{Model and objective}
\label{sec:setting:model}

\paragraph{M/G/$n$ System.}
There are $n$ identical servers of speed one. Jobs arrive as a Poisson process of rate $\lambda$, and the size of a job, namely the amount of work it requires, is drawn independently from the distribution of a random variable $S$. The \emph{response time} of a job is the time from its arrival to its departure.

\paragraph{Heavy-tailed job sizes.}
The job size distribution is \emph{regularly varying} with index $\alpha>1$: For every $k>0$,
\begin{equation}
    \lim_{t\to\infty}\frac{\P{S>kt}}{\P{S>t}}=k^{-\alpha}.
    \label{eq:regular varying definition}
\end{equation}
In words, the tail of $S$ is a power law up to a slowly varying factor, as for the Pareto distribution, and $\alpha>1$ makes the mean size $\E{S}$ finite. Regular variation is the standard formalization of heavy-tailed job sizes in the tail-optimality literature \citep{wierman2012tail,nuyens2008preventing,yu2024strongly,split}. 

\paragraph{Fractional server allocation.}
We allow the capacity of a server to be divided among several jobs. At each time $t$, a scheduling policy assigns every job $j$ in the system a service rate $v_j(t)\in[0,1]$, and the rates are constrained only by the total capacity of the servers: 
\begin{equation}
    \sum_j v_j(t)\le n.
    \label{eq:rate constraint}
\end{equation}
The remaining size of job $j$ decreases at rate $v_j(t)$. In practice, a fractional rate is realized by time sharing a server between two jobs at a fine granularity.
 
\paragraph{Preemptible jobs.}
We assume that jobs are preemptible without loss of work. Preemption is indispensable under heavy-tailed job sizes. Once a huge job has started, a non-preemptive policy cannot prevent the small jobs that arrive during its service from waiting for it. In the single-server setting, no non-preemptive policy avoids the resulting large delays \citep{anantharam1988delays,boxma2007tails}.

\paragraph{Known job sizes.}
The policy knows the size of every job upon its arrival.

\paragraph{Load and spare capacity.}
Let
\[
    r:=\lambda\E{S},\qquad
    \rho:=\frac{r}{n},\qquad
    \delta:=n-r=n(1-\rho).
\]
The \emph{resource requirement} $r$ is the mean amount of work arriving per unit time, measured in servers; the \emph{load} $\rho$ is the same quantity as a fraction of the $n$ servers; and the \emph{spare capacity} $\delta$ is the capacity that remains after the arriving work is served. We assume $r<n$, or equivalently $\delta>0$, so that every work-conserving policy makes the system stable \citep{harchol2013performance}. We also write
\begin{equation}
    r^{>x}:=\lambda\E{S\cdot\one\{S>x\}}
    \label{eq:r_x}
\end{equation}
for the part of the resource requirement contributed by jobs larger than $x$. The results of this paper concern the high-load regime $\rho\ge\frac{n-1}{n}$, that is, $\delta\le1$.

\paragraph{Objective.}
For a policy $\pi$, let $T^\pi$ denote the response time of a job when the system is in its stationary regime. The objective is the asymptotic decay of $\P{T^\pi>t}$ as $t\to\infty$. A policy $\pi$ is \emph{weakly tail optimal} if there exists a constant $c\ge1$ such that
\begin{equation}
    \sup_{\pi'}\limsup_{t\to\infty}\frac{\P{T^\pi>t}}{\P{T^{\pi'}>t}}=c,
    \label{eq:asym tail}
\end{equation}
where the supremum is over all policies $\pi'$ in the model above, and \emph{strongly tail optimal} if additionally $c=1$. 

\subsection{Preliminaries: the super-server lower bound}
\label{sec:theory:lowerbound}

\paragraph{Super-server dominance.}
Consider a single server of speed $n$ that faces the same arrivals and the same job sizes as the $n$-server system; we call it the \emph{super server}. The super server can do everything that the $n$ servers can do \citep{harchol2013performance}. Hence every response time tail that the $n$-server system can achieve is also achieved by some super-server policy, and any lower bound that holds for every super-server policy is a lower bound for every $n$-server policy as well.

\paragraph{The single-server tail.} 
The super server is a single server, and, for a single server, the best possible tail is known. Consider a single server of speed one with load $\rho<1$ and regularly varying job sizes. Then no policy has a response time tail lighter than $\P{S>(1-\rho)t}$, and SRPT attains this tail \citep{wierman2012tail,nuyens2008preventing,yu2024strongly}: For every single-server policy $\pi'$,
\begin{equation}
    \liminf_{t\to\infty}\frac{\P{T^{\pi'}>t}}{\P{S>(1-\rho)t}}\ge1,
    \qquad
    \lim_{t\to\infty}\frac{\P{T^{\mathrm{SRPT}}>t}}{\P{S>(1-\rho)t}}=1.
    \label{eq:single server heavy tail}
\end{equation}

\paragraph{The super-server lower bound.}
Measured in units of its own speed, the super server is a single server of speed one that sees jobs of size $S/n$ and has load $\lambda\E{S/n}=\rho<1$. Applying~\eqref{eq:single server heavy tail} with $S/n$ in place of $S$, every super-server policy $\hat\pi$ satisfies
\begin{align*}
    &\liminf_{t\to\infty}\frac{\P{T^{\hat\pi}>t}}{\P{S/n>(1-\rho)t}}\ge1,\\
    &\text{where}\quad
    \P{S/n>(1-\rho)t}=\P{S>n(1-\rho)t}=\P{S>\delta t}.
\end{align*}
By the super-server dominance, the same lower bound also holds for every policy $\pi$ of the $n$-server system,
\begin{equation}
    \liminf_{t\to\infty}\frac{\P{T^\pi>t}}{\P{S>\delta t}}\ge1.
    \label{eq:lower bound super server}
\end{equation}
We call~\eqref{eq:lower bound super server} the \emph{super-server lower bound}. Theorem~\ref{thm:frac:optimality} shows that \DynoSPLIT attains~\eqref{eq:lower bound super server} at every load in the high-load regime.

% !TEX root = ../main.tex
\section{\DynoSPLIT: Policy Design and Proof Roadmap}
\label{sec:frac:roadmap}

In this section, we provide an overview of the policy design and the proof roadmap.
Section~\ref{sec:frac:policy} recalls the idea of \DynoSPLIT from
Section~\ref{sec:introduction} and formally defines the policy, and
Section~\ref{sec:frac:three-regimes} gives the roadmap for its tail analysis.
The detailed proof appears in Section~\ref{sec:frac:proof}.

\subsection{The \DynoSPLIT policy} 
\label{sec:frac:policy}
% Alan (PDF p.3, l.140, on the deleted 3.1 sentence "No fixed margin \varepsilon serves both."): "serves" -> "can satisfy".
% Alan (PDF p.4, l.167, deleted 3.1 paragraph): the stray "6" in "difference6" crossed out. 

\DynoSPLIT serves the largest job in the system, which we call the
\emph{anchor}, at rate \(\delta\) minus a margin that depends on the
anchor's size, and serves the other jobs in SRPT order with the remaining
capacity.  Recall the dilemma of Section~\ref{sec:intro:dilemma}: A huge job
must be served at a rate approaching the spare capacity \(\delta\), or the
response time tail loses a constant factor; simultaneously, the other jobs need a
positive share of the same spare capacity, or their work accumulates for as
long as the huge job is present.  To meet both requirements, \DynoSPLIT uses
a margin that is strictly positive for every finite anchor size and vanishes
as the anchor size grows.  Here ``largest'' is meant roughly;
precisely, the anchor is chosen as follows:
\begin{definition}[Anchor]
    \label{def:frac:anchor}
    A job arriving to an empty system becomes the \emph{anchor}.  The anchor
    keeps its role until it completes.  When the anchor completes, the job
    with the largest \emph{original} size among those present, if any,
    becomes the new anchor; ties are broken by arrival order.
\end{definition}

We now formalize the size-dependent margin and the resulting service
allocation.  For an anchor of original size \(z\), let \(\varepsilon(z)\),
the \emph{capacity margin}, denote the capacity withheld from the anchor
and given to the non-anchor jobs.  The anchor's service rate is
\begin{equation}
    \gamma(z):=\delta-\varepsilon(z).
    \label{eq:frac:gamma-definition}
\end{equation}
% Alan (PDF p.4, l.178): \epsilon(z) circled, \varepsilon written above (already applied).
The margin $\varepsilon(z)$ must stay strictly between zero and $\delta$ at every finite anchor size $z$, so that both the anchor and the non-anchor jobs receive positive capacity, and vanish in the large-size limit to ensure strong optimality. Formally, we require
\begin{equation}
    0<\varepsilon(z)<\delta,
    \qquad
    \lim_{z\to\infty}\varepsilon(z)=0.
    \label{eq:frac:margin-basic}
\end{equation}
%  We also bound it by
% \(\delta/2\), which guarantees the anchor a service rate of at least
% \(\delta/2\).  Formally, we require
% \begin{equation}
%     0<\varepsilon(z)\le \frac{\delta}{2},
%     \qquad
%     \varepsilon(z)\ \text{is nonincreasing},
%     \qquad
%     \varepsilon(z)\longrightarrow0.
%     \label{eq:frac:margin-basic}
% \end{equation}
Here
\(\delta\) is fixed by the system load; it is \(\varepsilon(z)\), not
\(\delta\), that vanishes with the anchor size.

Note that \DynoSPLIT serves \(n+1\)
jobs simultaneously.  Whenever at least
\(n+1\) jobs are present, the \(n-1\) smallest
non-anchor jobs each receive a full server, while one additional unit of
capacity is shared between the anchor and one more non-anchor job: 
\[
    \underbrace{\overbrace{1+\cdots+1}^{n-1\text{ smallest jobs}}}_{n-1\text{ servers}}
    +\underbrace{\overbrace{\bigl(1-\gamma(z)\bigr)}^{n\text{th smallest job}}
      +\overbrace{\gamma(z)}^{\text{anchor}}}_{\text{one server}}
    =n.
\]
This allocation is allowed by the fractional-server model in
Section~\ref{sec:setting}; Figure~\ref{fig:frac:allocation} illustrates it.

\begin{figure}[t]
\FIGURE
{\includegraphics[width=0.72\textwidth]{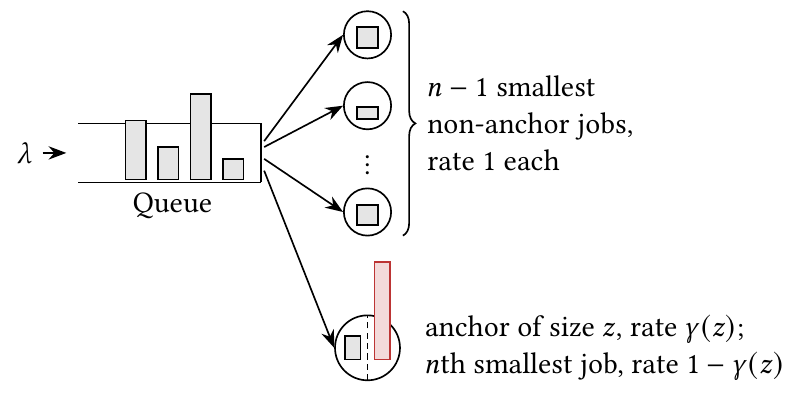}}
{Service allocation under \DynoSPLIT when more than \(n\) jobs are present.\label{fig:frac:allocation}}
{}
% Alt text: A queue feeds n servers drawn as circles; the last server is shared by the nth smallest job and the anchor.
\end{figure}

\begin{definition}[\DynoSPLIT with capacity margin \(\varepsilon\)]
\label{def:frac:policy}
Let the anchor be as in Definition~\ref{def:frac:anchor}.
If at most \(n\) jobs are present, every job is served at rate one.  If more
than \(n\) jobs are present and the anchor has original size \(z\), then: 
\begin{enumerate}\renewcommand{\theenumi}{\roman{enumi}}\renewcommand{\labelenumi}{\textup{(\theenumi)}}
    \item the anchor is served at rate \(\gamma(z)\);
    \item the \(n-1\) non-anchor jobs with the smallest remaining sizes are
          each served at rate one;
    \item the non-anchor job with the nth smallest remaining size is served
          at rate \(1-\gamma(z)\).
\end{enumerate}
All other jobs receive rate zero.  Ties in the non-anchor SRPT order are
broken by arrival order.
\end{definition}

The anchor is selected using original size, whereas service priority among
the other jobs uses remaining size.  The anchor role is
non-preemptive: A job larger than the current anchor waits on the non-anchor
side until the current anchor completes.  Non-preemption is adopted only for ease of
analysis.  It guarantees that, once a job becomes the anchor, its remaining
service is determined by its own size alone.

The allocation leaves the non-anchor jobs exactly \(\varepsilon(z)\) units
of spare capacity.  Indeed, when enough non-anchor jobs are present to use
their entire allocation, their aggregate service rate is
\begin{equation}
    (n-1)+\bigl(1-\gamma(z)\bigr)
    =n-\delta+\varepsilon(z)
    =r+\varepsilon(z).
    \label{eq:frac:nonanchor-capacity}
\end{equation}

The policy is also work conserving: Whenever at least \(n\) jobs are
present, it uses all \(n\) units of capacity; when fewer are present, every
job is served at rate one.  Since \(r<n\), the system is therefore stable
\citep{harchol2013performance}, and below we consider it in steady state.

The main result of this paper is the following theorem.

\begin{theorem}[Strong tail optimality of \DynoSPLIT]
\label{thm:frac:optimality}
Suppose \(n\ge2\), \((n-1)/n\le\rho<1\), and the job-size distribution has a
regularly varying tail with tail index \(\alpha>1\).  For the capacity
margin constructed in Section~\ref{sec:frac:combining},
\[
    \lim_{t\to\infty}
    \frac{\P{T^{\DynoSPLIT}>t}}{\P{S>\delta t}}
    =1,
    \qquad\text{where}\quad \delta:=n(1-\rho).
\]
Consequently, \DynoSPLIT is strongly tail optimal.
\end{theorem}

\subsection{Proof Roadmap}
\label{sec:frac:three-regimes}

The super-server lower bound~\eqref{eq:lower bound super server} already
gives the required lower bound, so the proof only needs a matching upper
bound.  We obtain it by a tagged-job argument: Consider a job arriving to
the stationary system, which we call the \emph{tagged job}, and let \(x\) denote its size; hence \(x\sim S\).  We write \(T_x^{\DynoSPLIT}\) for the response time of the tagged job.  Since \(x\) is distributed as \(S\), we have \(\P{T^{\DynoSPLIT}>t}=\P{T_x^{\DynoSPLIT}>t}\).  We now split the event \(\{T_x^{\DynoSPLIT}>t\}\) according to the size \(x\).

For each response-time threshold \(t\), we introduce two tagged-job size
thresholds.  The lower threshold \(\ell_t\) separates small jobs from
large jobs, and the upper threshold \(h_t\) separates large jobs from huge
jobs.  The response-time tail then satisfies 
\begin{align}
    \P{T^{\DynoSPLIT}>t} = \P{T_x^{\DynoSPLIT}>t}
    &\le \P{T_{x}^{\DynoSPLIT}>t,\ x\le\ell_t}
       \notag\\
    &\quad+\P{T_{x}^{\DynoSPLIT}>t,\ \ell_t<x\le h_t}
       +\P{x>h_t}.
    \label{eq:frac:three-way-split}
\end{align}
For brevity, we write  
\[
    \Fbar(x):=\P{S>x}.  
\]
The proof will choose the thresholds so that, as \(t\to\infty\),
\begin{align}
    \P{T_{x}^{\DynoSPLIT}>t,\ x\le\ell_t}
       &=o\bigl(\Fbar(\delta t)\bigr),
       \label{eq:frac:roadmap-small}\\
    \P{T_{x}^{\DynoSPLIT}>t,\ \ell_t<x\le h_t}
       &=o\bigl(\Fbar(\delta t)\bigr),
       \label{eq:frac:roadmap-large}\\
    \Fbar(h_t)
       &=\bigl(1+o(1)\bigr)\Fbar(\delta t).
       \label{eq:frac:roadmap-huge}
\end{align}
After division by the lower bound, \(\Fbar(\delta t)\), the three terms on the right side of
\eqref{eq:frac:three-way-split} are bounded asymptotically by \(0\), \(0\),
and \(1\).  This gives the required equality in
Theorem~\ref{thm:frac:optimality}, which proves the strong optimality.
The breakdown into small, large, and huge jobs follows the analysis of
\SPLIT in \citet{split}; however, the proof, especially the small-job argument, is far more complex and requires new ideas.

\begin{figure}[t]
\FIGURE
{\includegraphics[width=0.95\textwidth]{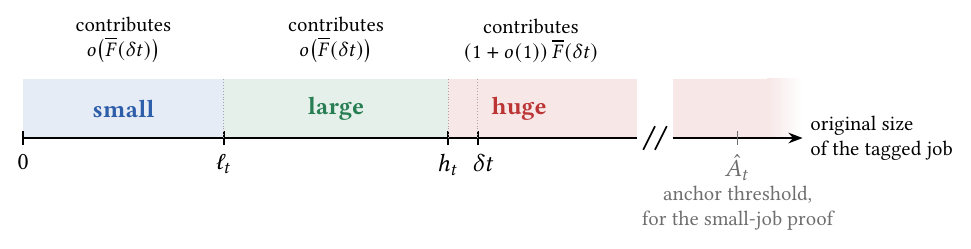}}
{Proof roadmap: the tagged job is small, large, or huge according to its original size.  Not to scale.\label{fig:frac:roadmap}}
{}
% Alt text: A size axis divided into small, large, and huge regions, with each region's contribution to the tail above it and the anchor threshold marked after an axis break.
\end{figure}
We now explain how small, large and huge jobs are handled respectively in the proof, and how the thresholds $\ell_t$ and $h_t$ are carefully chosen. Figure~\ref{fig:frac:roadmap} summarizes the roadmap.

\paragraph{Small jobs.}
In the analysis of \SPLIT \citep{split}, a job is called ``small'' if its size is at
most \(t/C\) for a constant \(C\).  This cutoff works in the lower-load
regime because the \(n-1\) SRPT servers keep the constant margin
\(\delta-1\) regardless of the size of the largest job.  Given a constant
margin, for a suitable constant \(C\), a job's response time exceeds \(C\) times its size
only with negligible probability.  

Two things change under \DynoSPLIT.
First, the margin of the non-anchor side depends on the system state,
through the size of the current anchor.  Second, this margin is arbitrarily
close to zero when the anchor is huge.  % Alan (PDF p.6, l.275, on the old sentence "We instead need a cutoff smaller than t/C. This scale explicitly accounts for the vanishing margin."):
%   "This" crossed out, "scale" -> "scaled", "to" inserted, "accounts" -> "account",
%   i.e. "..., scaled to explicitly account for the vanishing margin".
No constant \(C\) works any more:
The small-job cutoff \(\ell_t\) has to be smaller than \(t/C\) for every
constant \(C\), by a factor that reflects how small the margin can get.

Our solution for overcoming the above two analytic difficulties is to introduce a
threshold, \(\hatA_t\), on the size of the anchor. If the anchor size is upper bounded by \(\hatA_t\), then we know that the non-anchor capacity has
margin at least \(\varepsilon(\hatA_t)\). This allows us to break the proof into two cases: (1) the rare case when the anchor size is larger than \(\hatA_t\), and (2) the common case when the non-anchor jobs have at least \(\varepsilon(\hatA_t)\) extra capacity.

More precisely, for the first case (the rare case), we need the threshold \(\hatA_t\) to be large enough so that the following three exceptional events are rare: First, a job larger than \(\hatA_t\) may already be present
when the tagged job arrives, either as the anchor or waiting to become the
anchor; second, such an anchor may have recently departed, but the system is still
draining the backlog accumulated during its lifetime; third, a new job
larger than \(\hatA_t\) may arrive during \((0,t]\).
Outside these three events, no job larger than \(\hatA_t\) is present
during \([0,t]\), and no backlog left by an earlier such job remains.  The
non-anchor capacity then has margin at least \(\varepsilon(\hatA_t)\).
In this common case, the small-job proof reduces the delay of the tagged job to two sources of interfering work: the work of smaller non-anchor jobs present on arrival, and the work of jobs that are smaller than the tagged job arriving during \((0,t]\). These are the only jobs that SRPT places ahead of the tagged job. We will apply concentration arguments to show that this interfering work is not likely to exceed the extra \(\varepsilon(\hatA_t)t\) units of service available over a horizon of length \(t\).

The choice of the threshold \(\hatA_t\) thus has competing requirements.  It must be large
enough that the three exceptional events above have probability
\(o(\Fbar(\delta t))\) as \(t\to\infty\).  It must also be small enough that
\(\varepsilon(\hatA_t)\) remains large enough to control small-job
delays.

% Once the anchor is controlled, the small-job proof reduces delay to two
% sources of interfering work.  The first is work already present whose
% remaining size allows it to run before the tagged job; the second is new
% interfering work, upper-bounded by capping each arrival at \(\ell_t\).  A
% concentration argument shows that these two sources are unlikely to exceed
% the extra \(\varepsilon(\hatA_t)t\) units of service available over a
% horizon of length \(t\).  Because this extra capacity vanishes relative to
% \(t\), the bound must retain its dependence on both
% \(\varepsilon(\hatA_t)\) and \(\ell_t\).

\paragraph{Large jobs.}
The upper threshold \(h_t\) is dictated by strong optimality.  For huge
jobs, we use only a trivial bound (see below), so they contribute up to
\(\Fbar(h_t)\) to the response time tail.  For this contribution to match
the super-server lower bound \(\Fbar(\delta t)\), we must have \(h_t\) be close to
\(\delta t\), namely \(h_t\sim\delta t\). 

A large tagged job, with size \(x\in(\ell_t,h_t]\), must therefore
finish within time \(t\), even when \(x\) is close to \(\delta t\).  Once
the tagged job becomes the anchor, non-preemption guarantees that it keeps
the role until it completes, and it is served at rate at least
\(\gamma(x)\), which approaches \(\delta\) as \(x\) grows.  Before that,
however, the tagged job must wait for the current anchor to complete,
again because the anchor role is non-preemptive.  We therefore choose
\(h_t\) slightly below \(\delta t\): Serving a large job at rate
\(\gamma(x)\) takes at most time \(t\) minus a slack of \(o(t)\), and this
slack is reserved for the wait. Because the slack is \(o(t)\) and
\(\gamma(x)\) approaches \(\delta\), this choice is compatible with
\(h_t\sim\delta t\).
 
A large job can thus
miss its deadline only if it collides with another large job. Specifically, the tagged large job only misses its deadline if (1) when it arrives, the current anchor is too large to complete within the $o(t)$ slack, or if (2) a
job larger than \(x\) arrives during the wait and takes the anchor role
first.  Both collisions are rare, and making these probabilities negligible pushes
% and the first is the dominant one.  Averaging over the tagged size, the
% large-job contribution is at most the probability that the tagged job is
% large, times the fraction of time during which a job larger than
% \(\sqrt{\ell_t}\) is present.  Making this product negligible pushes 
% # The above sentences are proof details, not crucial.
\(\ell_t\) upward, in direct tension with the small-job proof, which needs
\(\ell_t\) to be sufficiently small.

\paragraph{Huge jobs.}
Huge jobs supply the unavoidable tail contribution. For a tagged job with
\(x>h_t\), we use only the trivial bound
\(\P{T_x^{\DynoSPLIT}>t,\ x>h_t}\le\P{x>h_t}=\Fbar(h_t)\).  Since
\(h_t\sim\delta t\), this contribution asymptotically matches the
super-server lower bound.

Section~\ref{sec:frac:proof} follows this roadmap in the same order.  Each
condition on the still-unspecified margin and thresholds is introduced only
when the small-, large-, or huge-job argument forces it. Only after all the lemmas are proved does
Theorem~\ref{thm:frac:explicit} choose explicit formulas for
\(\varepsilon\), \(\hatA_t\), \(\ell_t\), and \(h_t\) and verify the
accumulated conditions (See equation \eqref{eq:frac:explicit-thresholds}).  
% !TEX root = ../main.tex
\section{Proof of Strong Tail Optimality}
\label{sec:frac:proof}

In this section, we prove Theorem~\ref{thm:frac:optimality} by following
the roadmap of Section~\ref{sec:frac:three-regimes}.
Section~\ref{sec:frac:small} proves the small-job
bound~\eqref{eq:frac:roadmap-small}, Section~\ref{sec:frac:large} proves the
large-job bound~\eqref{eq:frac:roadmap-large}, and
Section~\ref{sec:frac:huge} proves the huge-job
bound~\eqref{eq:frac:roadmap-huge}.  Section~\ref{sec:frac:combining}
chooses the capacity margin and the thresholds and combines the three
bounds into the proof of Theorem~\ref{thm:frac:optimality}.  Throughout
the first three subsections, the capacity margin \(\varepsilon\) and the
thresholds \(\hatA_t\), \(\ell_t\), and \(h_t\) remain unspecified.

We consider the system in steady state.  The tagged job of Section~\ref{sec:frac:three-regimes}, with
size \(x\) and response time \(T_x^{\DynoSPLIT}\), arrives at time zero.
% Auxiliary processes introduced from time zero need not themselves be stationary. # this sentence is very confusing: auxiliary process? what are the authors talking about?

Unless another variable is indicated, every limit and asymptotic relation
in this section and in Appendix~\ref{app:frac:rv}, including \(\to\), \(\sim\),
\(\ll\), and \(o(\cdot)\), is taken as \(t\to\infty\).

Throughout this section, the capacity margin satisfies
\eqref{eq:frac:margin-basic} together with the following two shape
requirements:
\begin{equation}
    \varepsilon(z)\le\frac{\delta}{2}
    \quad\text{for every \(z\ge0\)},
    \qquad
    \varepsilon\ \text{is nonincreasing}.
    \label{eq:frac:condition-margin-shape}
\end{equation}
The first requirement guarantees the anchor a service rate of at least
\(\delta/2\) whatever its size.  Three places in the proof need such a uniform
floor: the bound on how often a large job is present in
Section~\ref{sec:frac:recovery}, the bound on the tagged job's own
progress in Section~\ref{sec:frac:capped-work}, and the bound on how long
a large job waits for the anchor role in Section~\ref{sec:frac:large}.  The fraction \(1/2\) is
adopted only for ease of analysis.\endnote{Any fixed fraction of
\(\delta\) would serve the same purpose and only changes the constants.}
The second requirement is what turns an upper bound on the anchor size into
a lower bound on the non-anchor margin: If the anchor has size at most
\(a\), the non-anchor jobs have margin at least \(\varepsilon(a)\).  This is
how the small-job proof obtains a uniform margin from a bound on the anchor.
Neither requirement says how quickly the margin vanishes; that is decided in
Section~\ref{sec:frac:combining}.

\subsection{Small jobs: a vanishing margin still suffices}
\label{sec:frac:small}

We now prove the small-job bound~\eqref{eq:frac:roadmap-small}.  Recall the
obstacle described in the roadmap: the margin available to the non-anchor
jobs depends on the current anchor, and it can be arbitrarily close to
zero.  The proof follows the roadmap.  Step~1
(Section~\ref{sec:frac:recovery}) shows that the three exceptional events of the
roadmap, namely a job larger than \(\hatA_t\) present on
arrival, a backlog left behind by such a job, and an arrival larger than
\(\hatA_t\) during the horizon, together have probability
\(o(\Fbar(\delta t))\).  Outside these events, the non-anchor side has a
uniform margin \(\varepsilon(\hatA_t)\) over the whole horizon.  Step~2
(Section~\ref{sec:frac:capped-work}) shows that, with this uniform margin,
the tagged job is late only if the two sources of interfering work from the
roadmap, namely the relevant work present on arrival and the capped work
arriving afterwards, deviate from their mean by a lot.  Step~3 (Section~\ref{sec:frac:concentration})
% shows by
% concentration that
% this is exponentially unlikely, with the vanishing margin
% kept in the exponent.  Figure~\ref{fig:frac:small-reduction} summarizes the
% chain of reductions.
finally applies concentration bounds on the above two terms of work to finish the proof.
% # The original sentence has several problems: "this" means what? Also, "the vanishing margin kept in the margin" also means nothing -- It is too detailed. Finally, the figure is not helpful at all. The figure does not establish any intuitive figures in readers' mind, and is far too technical since it includes so many equations.

% \begin{figure}[t]
%     \centering
%     \includegraphics[width=0.95\linewidth]{figs/small_job_reduction.pdf}
%     \caption{The chain of reductions in the small-job proof.  Each box is
%     an event containing the one above it, up to the bad events discarded in
%     Step~1, and each arrow names the step or lemma that proves the
%     inclusion.  The final two deviations are controlled by
%     Lemma~\ref{lem:frac:concentration}.}
%     \Description{Five stacked boxes connected by downward arrows.  The top
%     box is the event that a small tagged job is late.  The arrows are
%     labeled by the steps of the small-job proof, and the bottom box is the
%     union of two deviation events.}
%     \label{fig:frac:small-reduction}
% \end{figure}

% \subsubsection{Step 1: Discard histories with a red anchor}
\subsubsection{Step 1: Discard exceptional events}
\label{sec:frac:recovery}
We first show that the three exceptional events above are rare.
The main result of this step is Corollary~\ref{cor:frac:good-horizon}.

Fix the horizon \(t\).  The roadmap controls the small-job margin through
the anchor threshold \(\hatA_t\): While the anchor has size at most
\(\hatA_t\), the non-anchor jobs have margin at least
\(\varepsilon(\hatA_t)\) by~\eqref{eq:frac:condition-margin-shape}.  The
value of \(\hatA_t\) remains unspecified.
Nevertheless, we should imagine \(\hatA_t\) to be very large. Specifically, we require
% For the estimates of this step
% to apply, and so that a small tagged job never exceeds the anchor threshold
% itself, we require
\begin{equation}
    \hatA_t\longrightarrow\infty,
    \qquad
    \ell_t\ll\hatA_t.
    % \quad\text{for all sufficiently large \(t\)}.
    \label{eq:frac:condition-cap}
\end{equation}

We begin with the first exceptional event: A job with size \(>\hatA_t\) is present when the tagged job arrives.  Let \(Z(s)\) be the largest original job size
present at time \(s\), where \(Z(s)=0\) when the system is empty. Recall that \(Z(0)\) represents the stationary distribution of $Z(s)$, which is what the tagged job sees when it arrives.
The following lemma bounds how often a job of size \(>a\) is present, for a general threshold \(a\); in
this first step we only use it for \(a=\hatA_t\), but the large-job proof
in Section~\ref{sec:frac:large} reuses it at a different threshold.
% Call a job
% \emph{red} if its original size exceeds a given threshold \(a\); the anchor
% exceeds \(a\) exactly when a red job is present.  The following bound is
% stated for a general threshold because the large-job proof reuses it at a
% different threshold.
% # There is no need to introduce the concept of red. It is not the core thing here. It is also confusing because a is gonna change.

\begin{lemma}
\label{lem:frac:anchor-occupancy}
For every \(a>0\),
\begin{equation}
    \P{Z(0)>a}
    \le \frac{4r^{>a}}{\delta}.
    \label{eq:frac:anchor-occupancy}
\end{equation}
where \(r^{>a}\) is the part of the resource requirement contributed by
jobs larger than \(a\), defined in~\eqref{eq:r_x}.
\end{lemma}

\begin{proof}
Let \(W_a(s)\) be the total remaining work in jobs with original
size exceeding \(a\) at time \(s\); this work enters the system at mean
rate \(r^{>a}\), and the events \(\{W_a(0)>0\}\) and \(\{Z(0)>a\}\) are
identical.

Consider a period during which \(W_a>0\).  It starts when a
job larger than \(a\) arrives while no such job is present, and it ends
when \(W_a\) returns to zero; Figure~\ref{fig:frac:large-work} illustrates
one such period.
The job that starts the period joins the non-anchor
side, and the current anchor, if any, has original size at most \(a\).
The current anchor is served at rate one when at most \(n\) jobs are
present, and at rate \(\gamma\ge\delta/2\) otherwise,
by~\eqref{eq:frac:condition-margin-shape}, where \(\delta\le1\), so the
current anchor completes within \(2a/\delta\) time units.  From
then on, whenever the anchor completes while \(W_a>0\), the largest job
present is larger than \(a\) and becomes the new anchor.  Hence, after an
initial delay of at most \(2a/\delta\), the anchor is a job larger than
\(a\) for the rest of the period, and \(W_a\) decreases at rate at least
\(\delta/2\).  Other jobs larger than \(a\) may also be served on the
non-anchor side, which only makes \(W_a\) decrease faster.  Let \(B\) be
the length of the period where \(W_a>0\), and let \(W\) be the
total work of jobs larger than \(a\) that arrive during this period, including the
job that starts it.  Then \(W>a\), and
\[
    B\le\frac{2a}{\delta}+\frac{2W}{\delta}
    <\frac{4W}{\delta}.
\]
Summing over the successive $W_a>0$ periods in \([0,T]\) and letting
\(T\to\infty\), the long-run fraction of time with \(W_a>0\) is at most
\(4/\delta\) times the long-run rate at which work of jobs larger than
\(a\) arrives, which is \(r^{>a}\).  Since we consider the system in steady
state, this long-run fraction of time equals
\(\P{W_a(0)>0}\), which proves \eqref{eq:frac:anchor-occupancy}.
\end{proof}
 
\begin{figure}[t]
\FIGURE
{\includegraphics[width=0.7\textwidth]{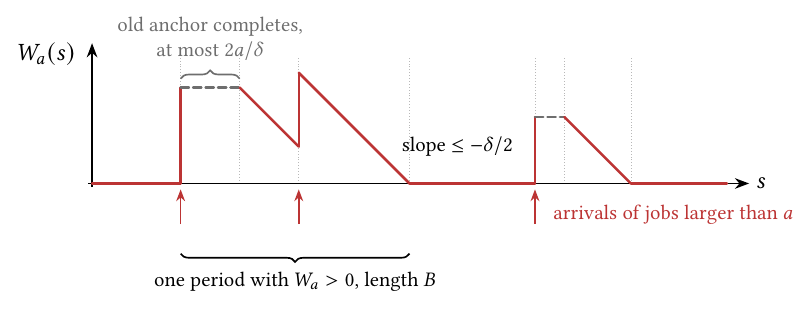}}
{One period during which \(W_a(s)>0\).\label{fig:frac:large-work}}
{}
% Alt text: A piecewise-linear curve that jumps at arrivals of large jobs, stays flat for a short dashed stretch, and then decreases until it reaches zero.
\end{figure}

% # It makes more sense to define red here, since now we only talk about Ahat.
The second exceptional event is subtler. Call a job \emph{red} if its original size exceeds \(\hatA_t\). A red anchor reduces the margin while it
is present, and the backlog accumulated during that time may remain after
the anchor departs.
%  Excluding only the event \(\{Z(0-)>\hatA_t\}\) is therefore not enough.
Hence, when a tagged job arrives at the system, even if it does not see a red anchor, it might be greatly delayed by the backlog accumulated before the last red anchor departs.
 We track the time the system needs, after a period
with a red anchor, before the surplus capacity of the following periods has
compensated for the lost margin. After this amount of time, the effect of a red anchor is cleared.

To make this precise, we first record the capacity of the non-anchor side.
When there are at least \(n\) non-anchor jobs, the capacity available to
them is at least
\begin{equation}
    \mu_{\mathrm{na}}(s)
    :=r+\varepsilon(Z(s)).
    \label{eq:frac:mu-na}
\end{equation}
Indeed, the anchor is one of the jobs present, so its original size
is at most \(Z(s)\), and \(\varepsilon\) is nonincreasing; the capacity
\(n-\gamma\) left to the non-anchor jobs by an anchor of size at most
\(Z(s)\) is therefore at least \(r+\varepsilon(Z(s))\).  We work with
this lower bound throughout.
% The word ``available'' matters here: if fewer than \(n\) non-anchor jobs
% are present, they may not use all of \(\mu_{\mathrm{na}}(s)\).

We call time \(s\) \emph{bad} if a red job is present at time \(s\), that is, if \(Z(s)>\hatA_t\).  We call time \(s\) \emph{good} if it is not bad.  During good time, the available
non-anchor capacity is at least \(r+\varepsilon(\hatA_t)\); during bad time,
it is at least \(r\).  We compare both with their midpoint, which we call the
\emph{benchmark rate} $c_t$: 
\begin{equation}
    m_t:=\frac{\varepsilon(\hatA_t)}{2},
    \qquad
    c_t:=r+m_t.
    \label{eq:frac:comparison-parameters}
\end{equation}
Good time then has a surplus of at least \(m_t\) above the benchmark rate,
and bad time has a deficit of at most \(m_t\) below it.  Hence one unit of
good time compensates for one unit of bad time, when both are measured
relative to \(c_t\).
Note that the benchmark rate still has a capacity margin $m_t$ over the service requirement $r$, which allows for the concentation arguments in Step 3.

The \emph{recovery debt} makes this accounting pathwise.  Write
\(B(s):=\one\{Z(s)>\hatA_t\}\) for the indicator of bad time, and, on a
two-sided stationary sample path, let
\begin{equation}
    D(s)
    :=\sup_{u\le s}\int_u^s\bigl(2B(v)-1\bigr)\,dv.
    \label{eq:frac:recovery-debt}
\end{equation}
As Figure~\ref{fig:frac:recovery-debt} illustrates, the debt
rises at rate one during bad time, which is shown in red, falls at rate one
during good time, shown in green, while it is positive, and stays at zero
once it has been repaid.  We call the system \emph{recovered} at time \(s\)
if \(D(s)=0\). Hence, a system can be good but not recovered, in which case we call it \emph{recovering}, and any recovered system is good.  In Figure~\ref{fig:frac:recovery-debt}, recovered periods are shaded in blue and recovering periods in green.

The following lemma bounds the fraction of time when the system is not recovered. From Figure \ref{fig:frac:recovery-debt}, it is intuitive to see that the fraction of unrecovered time is at most double the fraction of bad time; this fact, together with Lemma \ref{lem:frac:anchor-occupancy}, completes the proof.

\begin{figure}[t]
\FIGURE
{\includegraphics[width=0.95\textwidth]{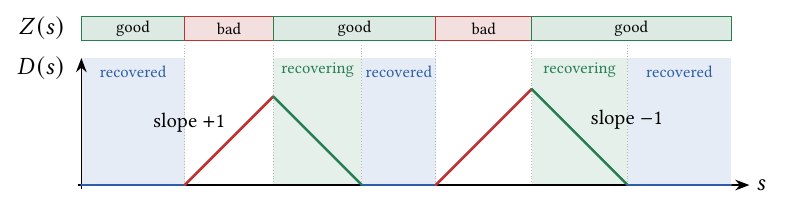}}
{Good and bad time, the recovery debt \(D(s)\), and the recovered and recovering periods.\label{fig:frac:recovery-debt}}
{}
% Alt text: A time line alternates between good and bad intervals; the debt rises during bad intervals and falls during good intervals. Periods with zero debt are shaded blue and labeled recovered; good periods with positive debt are shaded green and labeled recovering.
\end{figure}

\begin{lemma}[Recovery debt and effective capacity]
\label{lem:frac:recovery}
Under condition~\eqref{eq:frac:condition-cap}, for all sufficiently large
\(t\), the stationary recovery debt is finite and satisfies
\begin{equation}
    \P{D(0)>0}
    \le \frac{8r^{>\hatA_t}}{\delta}.
    \label{eq:frac:recovery-probability}
\end{equation}
Moreover, if \(D(s)=0\), then the capacity bound \(\mu_{\mathrm{na}}\)
over every backward interval ending at \(s\) satisfies
\begin{equation}
    \int_u^s\mu_{\mathrm{na}}(v)\,dv
    \ge c_t(s-u)
    \qquad\text{for every \(u\le s\)}.
    \label{eq:frac:backward-capacity}
\end{equation}
\end{lemma}

That is, in a recovered state the non-anchor jobs have received at least the
benchmark rate over every window ending at the present time.

\begin{proof}
Lemma~\ref{lem:frac:anchor-occupancy} at \(a=\hatA_t\), together with
\(\hatA_t\to\infty\) and \(r^{>a}\to0\) as \(a\to\infty\), gives
\(\P{B(0)=1}<1/2\) for all sufficiently large \(t\).  This is what makes the debt finite and \(D(s)<\infty\) almost surely.  This is the Loynes
construction~\citep{loynes1962stability}, and \(D\) is a stationary
functional of the stationary sample path.

We now compare the time the debt spends rising with the time it spends
falling.  Over any interval \([0,T]\), the debt rises at rate one during
bad time and falls at rate one during good time while positive, so
\[
    D(T)-D(0)
    =\int_0^T B(v)\,dv
     -\int_0^T\one\{B(v)=0,\ D(v)>0\}\,dv.
\]
Since \(D(T)\ge0\), the time spent falling is at most the time spent
rising plus \(D(0)\).  Now divide by \(T\) and let \(T\to\infty\).  The two
integrands are bounded stationary processes, so the ergodic theorem gives
their time averages, and \(D(0)/T\to0\) because \(D(0)<\infty\).  Hence we
have that 
\[
    \P{B(0)=0,\ D(0)>0}
    \le\P{B(0)=1}.
\]
Adding \(\P{B(0)=1,\ D(0)>0}\le\P{B(0)=1}\) to both sides and applying
Lemma~\ref{lem:frac:anchor-occupancy} gives
\[
    \P{D(0)>0}
        \le 2\P{B(0)=1}
    \le \frac{8r^{>\hatA_t}}{\delta}.
\]

It remains to connect recovery debt with capacity.  By the definition of
the benchmark rate and the monotonicity of \(\varepsilon\),
\[
    c_t-\mu_{\mathrm{na}}(v)
    =m_t-\varepsilon(Z(v))
    \le
    \begin{cases}
       -m_t, & Z(v)\le \hatA_t,\\
        m_t, & Z(v)>\hatA_t.
    \end{cases}
\]
Equivalently,
\begin{equation}
    c_t-\mu_{\mathrm{na}}(v)
    \le m_t\bigl(2B(v)-1\bigr).
    \label{eq:frac:pointwise-capacity}
\end{equation}
If \(D(s)=0\), definition~\eqref{eq:frac:recovery-debt} says that
\(\int_u^s(2B(v)-1)\,dv\le0\) for every \(u\le s\).  Integrating both sides of~\eqref{eq:frac:pointwise-capacity} over \((u,s]\) when $D(s)=0$ yields
\[
    c_t(s-u)-\int_u^s\mu_{\mathrm{na}}(v)\,dv\le0,
\]
which is exactly~\eqref{eq:frac:backward-capacity}.
\end{proof}
 
We can now exclude all three exceptional events at once: a red job present when the tagged job arrives, a backlog left behind by an earlier red job, and a red job arriving during \((0,t]\).  In essence, none of them occurs if the system is recovered throughout the horizon \([0,t]\).  Formally, let
\begin{equation}
    \mathcal G_t:=
    \left\{
        D(0-)=0,\quad
        \substack{\text{no arrival of size greater than }\hatA_t\\
                  \text{occurs in }(0,t]}
    \right\}.
    \label{eq:frac:good-horizon}
\end{equation}
% On \(\mathcal G_t\), every job present at time \(0-\) or arriving during
% \((0,t]\) has original size at most \(\hatA_t\).  Since a small tagged
% job also satisfies \(S\le\ell_t\le\hatA_t\)
% by~\eqref{eq:frac:condition-cap}, the anchor never exceeds \(\hatA_t\)
% on \([0,t]\).

\begin{corollary}[Exceptional events are rare]
\label{cor:frac:good-horizon}
Under condition~\eqref{eq:frac:condition-cap}, for all sufficiently large
\(t\),
\begin{equation}
    \P{\mathcal G_t^c}
    \le \frac{8r^{>\hatA_t}}{\delta}
       +\lambda t\Fbar(\hatA_t).
    \label{eq:frac:good-horizon-bound}
\end{equation}
\end{corollary}

\begin{proof}
By PASTA, the tagged arrival sees the stationary probabilities in
Lemma~\ref{lem:frac:recovery}.  A union
bound gives
\begin{align*}
    \P{\mathcal G_t^c}
    &\le \P{D(0-)>0}
       +\P{\substack{\text{an arrival larger than }\hatA_t\\
                     \text{occurs in }(0,t]}}\\
        &\le \frac{8r^{>\hatA_t}}{\delta}
       +\lambda t\Fbar(\hatA_t).
\end{align*}
For the last term, the number of such arrivals is Poisson with mean
\(\lambda t\Fbar(\hatA_t)\), and
\(1-e^{-y}\le y\).  This proves~\eqref{eq:frac:good-horizon-bound}.
\end{proof}

The bound~\eqref{eq:frac:good-horizon-bound} reveals the first asymptotic
requirement on \(\hatA_t\), in the direction the roadmap announced: the
threshold must be large enough that the exceptional events are negligible on the scale of
the target tail. Formally,
\begin{equation}
    r^{>\hatA_t}
    +t\Fbar(\hatA_t)
    =o\bigl(\Fbar(\delta t)\bigr).
    \label{eq:frac:condition-anchor}
\end{equation}
Under this requirement,
\begin{equation}
    \P{\mathcal G_t^c}
    =o\bigl(\Fbar(\delta t)\bigr).
    \label{eq:frac:good-horizon-probability}
\end{equation}
The event \(\mathcal G_t\) hands two facts to the next step: the system is recovered
at time zero, so the backlog present then is controlled with high probability, and the anchor stays below \(\hatA_t\) over the horizon, so the non-anchor side has the
uniform margin \(2m_t\) throughout.

\subsubsection{Step 2: Reduce response time to capped relevant work}
\label{sec:frac:capped-work}
% # The previous starting sentence is obviously incorrect.
% We now carry out the second reduction of the roadmap and turn the uniform
% margin into a bound on the tagged job's response time.
We now bound the tagged job's response time by the relevant work it sees and the relevant work that arrives; here relevant work is the work composed of jobs with remaining size smaller than the tagged job, to be defined in the next paragraph. Lemma \ref{lem:frac:tagged-work-bound} formalizes this argument, and Lemma \ref{lem:frac:work-comparison} uses the fact that the system is recovered to bound the relevant work the tagged job sees. Finally, we get the main lemma of this step 2, Lemma \ref{lem:frac:small-pathwise}.

% # Need to start with DynoSPLIT
The starting point
is that under \DynoSPLIT, because most servers do SRPT, a small job only sees small work.  A tagged job of size
at most \(\ell_t\) can be delayed on the non-anchor side only by non-anchor
jobs whose remaining size is at most \(\ell_t\), because those are the only
jobs that SRPT can place ahead of it. We are going to prove that the capacity that \DynoSPLIT gives to the non-anchor side is enough to guarantee a small response time of the tagged job.

The
right quantity to track is therefore the \emph{relevant work}.  Let
\(R_j(s)\) be job \(j\)'s remaining size at time \(s\), and let
\(j^\star(s)\) be the anchor.  Define
\begin{equation}
    W^{\mathrm{rel}}(s)
    :=\sum_{\substack{j\text{ present at }s\\j\ne j^\star(s)}}
       R_j(s)\one\{R_j(s)\le \ell_t\}.
    \label{eq:frac:relevant-work}
\end{equation}
This quantity includes all non-anchor work that can run ahead of a tagged
job of size at most \(\ell_t\).  It may also include some work that in fact
runs behind the tagged job, which is harmless for an upper bound.

New work matters through the same lens.  A job arriving during the horizon
can run ahead of the tagged job only during the arriving job's last \(\ell_t\) units of
service, so each arrival can create at most its size, capped at \(\ell_t\),
as new relevant work.  For job \(i\), let \(S_i\) and \(\tau_i\) denote its
original size and arrival time, and define the \emph{capped arriving work}
over \((u,v]\) and its mean input rate by
\begin{equation}
    \mathcal A(u,v]
    :=\sum_{i:\,u<\tau_i\le v}(S_i\wedge \ell_t),
    \qquad
    r_{\ell_t}:=\lambda\E{S\wedge \ell_t},
    \label{eq:frac:capped-arrivals}
\end{equation}
where $X\wedge Y:= \min\{X,Y\}.$
Capping does not describe a modified scheduling policy; it only bounds the
new work that can acquire priority over the tagged job.

The non-anchor capacity varies with the anchor, which makes the relevant
work hard to compare with a standard queue directly.  The recovery lemma, Lemma~\ref{lem:frac:recovery}, 
lets us replace the varying capacity by the constant benchmark rate
\(c_t\).  We therefore define
\begin{equation}
    U(s)
    :=\sup_{u\le s}
       \bigl\{\mathcal A(u,s]-c_t(s-u)\bigr\}.
    \label{eq:frac:comparison-workload}
\end{equation}
This is the stationary workload, constructed from the infinite past, of a
single server with rate \(c_t\) fed by the capped compound-Poisson input.
It is finite because \(c_t>r\ge r_{\ell_t}\), by the usual Loynes
construction~\citep{loynes1962stability}.  At the tagged arrival, \(U(0-)\)
denotes the value just before the tagged job enters.

The first lemma controls the relevant work present when the tagged job
arrives.

\begin{lemma}[Relevant-work comparison]
\label{lem:frac:work-comparison}
On \(\{D(s)=0\}\),
\begin{equation}
    W^{\mathrm{rel}}(s)
    \le (n+1)\ell_t+U(s).
    \label{eq:frac:work-comparison}
\end{equation}
\end{lemma}

That is, in a recovered state, the relevant work is at most the backlog of
the constant-rate comparison queue plus the work of \(n+1\) capped jobs.

\begin{proof}
    % # It is important to have first and second to guide the readers.
Relevant work can enter in two ways. First, an arrival increases it by at most
\(S_i\wedge \ell_t\); an arrival to an empty system becomes the anchor
and adds nothing.  Second, a non-anchor job can enter the relevant set when its
remaining size crosses \(\ell_t\). Such a crossing can occur only while
fewer than \(n\) relevant non-anchor jobs are present; otherwise all
available non-anchor capacity is assigned to jobs whose remaining sizes are
already at most \(\ell_t\).  Immediately after a crossing, therefore,
\[
    W^{\mathrm{rel}}\le n\ell_t.
\]

The desired bound is immediate when
\(W^{\mathrm{rel}}(s)\le(n+1)\ell_t\).  Otherwise, trace backward to the
start \(u\) of the current excursion above \(n\ell_t\).  An $\ell_t$ threshold down-crossing cannot start this excursion,
so it starts with an arrival and contains at most \((n+1)\ell_t\) work
immediately afterward.  Throughout the rest of the excursion, more than
\(n\) relevant non-anchor jobs are present.  They use all of the available
non-anchor capacity, which is at least \(\mu_{\mathrm{na}}(v)\)
by~\eqref{eq:frac:mu-na}, and hence
\[
    W^{\mathrm{rel}}(s)
    \le (n+1)\ell_t+\mathcal A(u,s]
              -\int_u^s\mu_{\mathrm{na}}(v)\,dv.
\]
Because \(D(s)=0\), Lemma~\ref{lem:frac:recovery} bounds the integral below
by \(c_t(s-u)\).  Taking the supremum over possible excursion starting times
gives~\eqref{eq:frac:work-comparison}.
\end{proof}

The second lemma converts the tagged job's own work and the work that can
run ahead of it into units of time.

\begin{lemma}[Tagged-job response-time bound]
\label{lem:frac:tagged-work-bound}
Suppose a tagged job of size \(x\le \ell_t\) arrives at time zero.  Then  
\begin{align}
&\left\{T_x^{\DynoSPLIT}>t,\
          \sup_{0\le s\le t}Z(s)\le \hatA_t\right\}
\notag\\
&\quad\subseteq
\left\{
    \frac{2}{\delta}x
    +\frac{W^{\mathrm{rel}}(0-)}{c_t}
    +\frac{\mathcal A(0,t]}{c_t}
    >t
\right\}.
\label{eq:frac:tagged-work-bound}
\end{align}
\end{lemma}

% That is, as long as the anchor stays below \(\hatA_t\), the tagged job can
% still be present at time \(t\) only if its own work, the relevant work it
% found on arrival, and the capped work that arrived afterwards together take
% more than \(t\) time units to serve at the guaranteed rates.

\begin{proof}
Let \(R(s)\) be the tagged job's remaining size at time \(s\).  Roughly
speaking, the non-anchor work that can run ahead of the tagged job must be
completed at rate at least \(c_t\); the only exception is when the tagged
job itself is being served, which happens at rate at least \(\delta/2\).
It remains to identify the work that can run ahead of the tagged job.
Under SRPT, a non-anchor job that is behind the tagged job never overtakes
it, so this work consists of two groups of jobs: the non-anchor jobs
present at time \(0-\) with remaining size at most \(\ell_t\), which are
exactly the jobs counted in \(W^{\mathrm{rel}}(0-)\), and the jobs arriving
after time zero with size at most \(\ell_t\).

To see this more precisely, let \(L(s)\) be the total remaining work at
time \(s\) of the jobs in the above two groups, and define the potential
\begin{equation}
    G(s):=\frac{2}{\delta}R(s)+\frac{1}{c_t}L(s).
    \label{eq:frac:tagged-potential}
\end{equation}
We verify three properties of \(G\).

\emph{Initial value.}
At time zero, \(L(0)=W^{\mathrm{rel}}(0-)\), so
\begin{equation}
    G(0)=\frac{2}{\delta}x+\frac{W^{\mathrm{rel}}(0-)}{c_t}.
    \label{eq:frac:credit-initial}
\end{equation}

\emph{Jumps.}
\(L\) increases only when a job of size \(S_i\le\ell_t\) arrives, and then
by exactly \(S_i\). Hence the total
increase of \(G\) over \((0,t]\) is at most
\[
    \frac{1}{c_t}\sum_{i:\,0<\tau_i\le t,\ S_i\le\ell_t}S_i
    \le\frac{\mathcal A(0,t]}{c_t}.
\]

\emph{Drift.}
Between arrivals, \(G\) decreases at rate at least one while the tagged
job is present.  Let \(q(s)\) be the tagged job's service rate.  There are
three cases.  First, if the tagged job is the anchor, then it is served at
rate one or at rate \(\gamma(x)\ge\delta/2\)
by~\eqref{eq:frac:condition-margin-shape}, so the first term of \(G\)
alone decreases at rate at least one.  Second, if the tagged job is a
non-anchor job with \(q(s)=1\), the same holds.  Third, suppose the tagged
job is a non-anchor job with \(q(s)<1\).  Then more than \(n\) jobs are
present and the tagged job is at least the \(n\)th in the non-anchor SRPT
order, so the jobs ahead of it receive all the non-anchor capacity that
the tagged job does not.  Writing \(p(s)\) for their total service rate,
\[
    p(s)+q(s)=n-\gamma\bigl(Z(s)\bigr)\ge\mu_{\mathrm{na}}(s)
\]
by~\eqref{eq:frac:mu-na}.  On the event on the left side
of~\eqref{eq:frac:tagged-work-bound}, \(Z(s)\le\hatA_t\), so, by the
monotonicity of \(\varepsilon\),
\[
    \mu_{\mathrm{na}}(s)
    \ge r+\varepsilon(\hatA_t)
    =c_t+m_t
    >c_t.
\]
We claim that every job ahead of the tagged job belongs to one of the two
groups counted in \(L\), so that \(L\) decreases at rate at least \(p(s)\) between arrivals.  Indeed, let
\(j\) be a non-anchor job at time \(s\) that belongs to neither group.
Since the anchor role is permanent, \(j\) was not the anchor at time
\(0-\); hence either \(j\) was present at time \(0-\) with
\(R_j(0-)>\ell_t\), or \(j\) arrived at some time \(\tau_j\in(0,s]\) with
\(S_j>\ell_t\).  In both cases, at the corresponding time
\(s_0\in\{0,\tau_j\}\) we have \(R_j(s_0)>\ell_t\ge x\ge R(s_0)\), so
\(j\) was behind the tagged job.  Both jobs are non-anchor jobs throughout
\([s_0,s]\), and the non-anchor service rates are nonincreasing along the
SRPT order, so \(R_j-R\) is nondecreasing as long as it is positive.
Hence \(R_j(s)>R(s)\), and \(j\) is still behind the tagged job at time
\(s\), which proves the claim.

Finally, the high-load assumption of Theorem~\ref{thm:frac:optimality}
gives \(r\ge n-1\ge1\) and \(\delta\le1\), so \(c_t>r\ge1>\delta/2\), that
is, \(2/\delta>1/c_t\): A unit of the tagged job's own service reduces
\(G\) at least as much as a unit of service to the jobs ahead of it.  Hence
we have that
\[
    -\frac{d}{ds}G(s)
    \ge \frac{2}{\delta}q(s)+\frac{1}{c_t}p(s)
    \ge \frac{p(s)+q(s)}{c_t}
    \ge \frac{\mu_{\mathrm{na}}(s)}{c_t}
    \ge1.
\]

If the tagged job is still present at time \(t\), then \(G(t)>0\).
Integrating the drift and adding the jumps, using
\eqref{eq:frac:credit-initial}, gives
\[
    0<G(t)
    \le
    \frac{2}{\delta}x
    +\frac{W^{\mathrm{rel}}(0-)}{c_t}
        +\frac{\mathcal A(0,t]}{c_t}-t,
\]
which proves the event inclusion.
\end{proof}

Combining the two lemmas on \(\mathcal G_t\) gives the following pathwise
reduction, which is the only output of this step that the concentration
argument needs.

\begin{lemma}[Small-job pathwise reduction]
\label{lem:frac:small-pathwise}
Let
\[
    K:=\frac{2(r+1)}{\delta}+n+1.
\]
Under condition~\eqref{eq:frac:condition-cap}, for all sufficiently large
\(t\),
\begin{align}
&\left\{T_{x}^{\DynoSPLIT}>t,\ x\le\ell_t\right\}\cap\mathcal G_t
\notag\\
&\quad\subseteq
\left\{
    K\ell_t+U(0-)
    +\mathcal A(0,t]
    >c_t t
\right\}.
\label{eq:frac:small-pathwise}
\end{align}
\end{lemma}

% That is, on a good horizon a small tagged job is late only if the
% constant-rate comparison queue, counting both its backlog at time zero and
% its capped input over \((0,t]\), holds more work than the benchmark rate can
% serve during the horizon, up to the deterministic term \(K\ell_t\).

\begin{proof}
On \(\mathcal G_t\), the condition \(D(0-)=0\) lets us apply
Lemma~\ref{lem:frac:work-comparison} at the tagged arrival:
\[
    W^{\mathrm{rel}}(0-)
    \le (n+1)\ell_t+U(0-).
\]
On \(\mathcal G_t\), no red job is present at time \(0-\), since the system is recovered then, and no red job arrives during \((0,t]\).  The tagged job is not red either, because \(x\le\ell_t\le\hatA_t\) by~\eqref{eq:frac:condition-cap}.  Hence the anchor never exceeds \(\hatA_t\) during
\([0,t]\), so
Lemma~\ref{lem:frac:tagged-work-bound} applies with \(x\le\ell_t\).
We multiply its work bound by \(c_t\) and substitute the comparison above.  The coefficient of \(\ell_t\) is then at most
\[
        \frac{2c_t}{\delta}+n+1.
\]
Because \(c_t\le r+1\), this coefficient is at most \(K\), proving
\eqref{eq:frac:small-pathwise}.
\end{proof}

\subsubsection{Step 3: Show that capped work does not consume the margin with very high probability}
\label{sec:frac:concentration}

We now carry out the last step of the roadmap:
% the capped work is unlikely
% to exceed the extra \(\varepsilon(\hatA_t)t=2m_tt\) units of service that
We use concentration bounds to prove that the event on the right-hand side of \eqref{eq:frac:small-pathwise} is rare.
% the margin provides over the horizon.
Subtracting the mean capped input
from the event on the right side of \eqref{eq:frac:small-pathwise}
identifies the two random deviations that remain:
\begin{align*}
&\left\{
    K\ell_t+U(0-)
    +\mathcal A(0,t]>c_t t
\right\}\\
&\quad=
\left\{
    U(0-)
    +\bigl(\mathcal A(0,t]-r_{\ell_t}t\bigr)
    >(c_t-r_{\ell_t})t-K\ell_t
\right\}.
\end{align*}
The first random term is backlog present at time zero; the second is the
amount by which future capped input exceeds its mean.  Since
\(r_{\ell_t}\le r\), we have \(c_t-r_{\ell_t}\ge m_t\), so the gap on the
right is at least \(m_tt-K\ell_t\).

We should imagine \(\ell_t\) as growing slower than $t$, so that \(K\ell_t\le m_tt/2\) for all large \(t\); the precise
requirement is stated before Proposition~\ref{prop:frac:small}.  The gap
on the right is then at least \(m_tt/2\), and, if the two deviations
together exceed the gap, one of them exceeds \(m_tt/4\).  Hence the work
event is contained in
\begin{align}
    &\left\{U(0-)>\frac{m_t t}{4}\right\}
    \notag\\[-2pt]
    &\qquad\cup
    \left\{
       \mathcal A(0,t]-r_{\ell_t}t
       >\frac{m_t t}{4}
    \right\}.
        \label{eq:frac:two-deviations}
\end{align}

The next lemma controls both deviations while retaining the vanishing margin
in the exponent.
% The fact that the centered input has mean zero is not
% sufficient for this: exponentiating it introduces a positive contribution
% from its random jumps, of the order of the second moment of a capped job
% size.  For that contribution to be negligible relative to the available
% margin, we require
% \begin{equation}
%     \frac{\E{(S\wedge\ell_t)^2}}{\ell_t}
%     =o(m_t)=o\bigl(\varepsilon(\hatA_t)\bigr).
%     \label{eq:frac:condition-small-moment}
% \end{equation}

\begin{lemma}[Capped work rarely exceeds the margin]
\label{lem:frac:concentration}
Assume the following condition:
\begin{equation}
    \frac{\E{(S\wedge\ell_t)^2}}{\ell_t}
    =o(m_t)=o\bigl(\varepsilon(\hatA_t)\bigr).
    \label{eq:frac:condition-small-moment}
\end{equation}
For all sufficiently
large \(t\),
\begin{align}
    \P{U(0-)>m_t t/4}
    &\le \exp\left\{-\frac{m_t t}{8\ell_t}\right\},
    \label{eq:frac:backlog-concentration}\\
    \P{\mathcal A(0,t]-r_{\ell_t}t>m_t t/4}
    &\le \exp\left\{-\frac{m_t t}{8\ell_t}\right\}.
    \label{eq:frac:arrival-concentration}
\end{align}
\end{lemma}

% That is, each deviation has probability exponentially small in the ratio of
% the margin buffer \(m_tt\) to the cap \(\ell_t\).  The cap appears because
% the largest possible jump of the capped input is \(\ell_t\), and the bound
% must hold uniformly over jumps of that size.

\begin{proof}
Let \(Y=S\wedge\ell_t\).  The capped input \(\mathcal A(0,s]\) is a
compound Poisson process with rate \(\lambda\) and jumps distributed as
\(Y\), so for every \(\theta>0\),
\[
    \E{e^{\theta\mathcal A(0,s]}}
    =\exp\bigl\{\lambda s\bigl(\E{e^{\theta Y}}-1\bigr)\bigr\}.
\]
Its mean is \(\E{\mathcal A(0,s]}=\lambda\E{Y}s=r_{\ell_t}s\).  Let
\[
    X(s):=\mathcal A(0,s]-r_{\ell_t}s,
    \qquad s\ge0,
\]
be the capped input with its mean removed.  Then \(X(0)=0\), and
\begin{equation}
    \log\E{e^{\theta X(s)}}
    =\lambda s\bigl(\E{e^{\theta Y}}-1\bigr)-\theta r_{\ell_t}s.
    \label{eq:frac:capped-mgf}
\end{equation}
Set \(\theta=1/\ell_t\).  Since \(0\le\theta Y\le1\), the elementary inequality
\(e^y\le1+y+(e/2)y^2\) gives 
\[
    \lambda s\bigl(\E{e^{\theta Y}}-1\bigr)
    \le\theta r_{\ell_t}s+\frac{e}{2}\lambda s\theta^2\E{Y^2},
\]%
and hence
\[
    \log\E{e^{\theta X(s)}}
        \le \frac{e\lambda s}{2}\theta^2\E{Y^2}.
\]
Condition~\eqref{eq:frac:condition-small-moment} says
\(\E{Y^2}/\ell_t=o(\varepsilon(\hatA_t))=o(m_t)\).  Thus, since $\frac{e \lambda}{2}$ is a constant independent of $t$, for all large \(t\), we have
\begin{equation}
    \E{e^{\theta X(s)}}
    \le \exp\left\{ \frac{e\lambda s}{2}\theta^2\E{Y^2} \right\} = \exp\left\{ \frac{e\lambda}{2} s\theta \frac{\E{Y^2}}{\ell_t} \right\} \le \exp\left\{\frac{m_t\theta s}{8}\right\}
    \qquad\text{for every \(s\ge0\)}.
        \label{eq:frac:uniform-mgf}
\end{equation}

Markov's inequality and~\eqref{eq:frac:uniform-mgf} yield
\[
    \P{X(t)>m_t t/4}
    \le \exp\{-\theta m_t t/4\}\E{e^{\theta X(t)}}
    \le \exp\left\{-\frac{m_t t}{8\ell_t}\right\},
\]
which proves~\eqref{eq:frac:arrival-concentration}.

For the stationary initial backlog, recall from
\eqref{eq:frac:comparison-workload} that
\(U(0-)=\sup_{u\le0}\{\mathcal A(u,0]-c_t(0-u)\}\).  The marked Poisson
process is stationary and reversible in time, so the capped input over
\((u,0]\) has the same law as that over \((0,-u]\).  Hence we have that
\begin{equation}
    U(0-)
    \stackrel{\mathrm d}{=}
    \sup_{s\ge0}\{\mathcal A(0,s]-c_ts\}
    =\sup_{s\ge0}\{X(s)-(c_t-r_{\ell_t})s\},
    \label{eq:frac:backlog-supremum}
\end{equation}
where the last equality uses the definition of \(X\).
Define
\[
    M(s):=\exp\{\theta[X(s)-(c_t-r_{\ell_t})s]\}.
\]
For \(0\le u\le v\), write
\(M(v)=M(u)\exp\{\theta[X(v)-X(u)]-\theta(c_t-r_{\ell_t})(v-u)\}\).
Conditionally on the path of \(X\) up to time \(u\), the factor \(M(u)\)
is fixed, and the increment \(X(v)-X(u)\) is independent of that path
and distributed as \(X(v-u)\).  Hence we have that
\begin{align*}
    \E{M(v)\mid X(w),\,w\le u}
    &=M(u)\,\E{e^{\theta X(v-u)}}\,e^{-\theta(c_t-r_{\ell_t})(v-u)}\\
    &\le M(u)\,e^{\theta m_t(v-u)/8}\,e^{-\theta(c_t-r_{\ell_t})(v-u)}
    &&\text{by~\eqref{eq:frac:uniform-mgf}}\\
    &\le M(u)
        &&\text{by \(c_t-r_{\ell_t}\ge m_t\)}.
\end{align*}
Hence \(M\) is a nonnegative supermartingale with \(M(0)=1\).  Doob's
maximal inequality states that a nonnegative supermartingale exceeds a
level \(a>0\) at some time in \([0,T]\) with probability at most
\(\E{M(0)}/a\) \citep{durrett2019probability}.  Fix \(z>0\) and
\(T>0\).  Since \(e^{\theta y}\) is increasing in \(y\), the process
\(X(s)-(c_t-r_{\ell_t})s\) exceeds \(z\) at some time in \([0,T]\)
exactly when \(M\) exceeds \(e^{\theta z}\) there.  Hence we have that
\begin{align*}
    \P{\sup_{0\le s\le T}\{X(s)-(c_t-r_{\ell_t})s\}>z}
    &=\P{\sup_{0\le s\le T}M(s)>e^{\theta z}}\\
    &\le \E{M(0)}\,e^{-\theta z}
    &&\text{by Doob's maximal inequality}\\
    &=e^{-\theta z}
        &&\text{by \(M(0)=1\)}.
\end{align*}
As \(T\to\infty\), the events on the left increase to
\(\{\sup_{s\ge0}\{X(s)-(c_t-r_{\ell_t})s\}>z\}\), so the same bound
holds for that event.  By \eqref{eq:frac:backlog-supremum}, for every
\(z>0\),
\[
    \P{U(0-)>z}
    \le e^{-\theta z}.
\]
Taking \(z=m_t t/4\) gives the stronger exponent \(-m_t t/(4\ell_t)\), and therefore
also proves~\eqref{eq:frac:backlog-concentration}.
\end{proof}

Lemma~\ref{lem:frac:concentration} bounds each deviation by
\(\exp\{-m_tt/(8\ell_t)\}\), whereas the target \(\Fbar(\delta t)\) decays
only polynomially in \(t\).  For the exponential to be
\(o(\Fbar(\delta t))\), the exponent must grow faster than every multiple
of \(\log t\).  Since \(m_t=\varepsilon(\hatA_t)/2\), we require
\begin{equation}
    \frac{\varepsilon(\hatA_t)t}
         {\ell_t\log t}
    \longrightarrow\infty.
    \label{eq:frac:condition-small-exponent}
\end{equation}
This is the requirement, announced in the roadmap, that pushes \(\ell_t\)
downward relative to the available margin.  It also implies
\(\ell_t=o(m_tt)\), as assumed in~\eqref{eq:frac:two-deviations}.   We can
now complete the first of the three bounds in the roadmap.

\begin{proposition}[Small-job contribution]
\label{prop:frac:small}
If
\eqref{eq:frac:condition-cap}, \eqref{eq:frac:condition-anchor},
\eqref{eq:frac:condition-small-exponent}, and
\eqref{eq:frac:condition-small-moment} hold, then
\begin{equation}
    \P{T_{x}^{\DynoSPLIT}>t,\ x\le\ell_t}
    =o\bigl(\Fbar(\delta t)\bigr).
    \label{eq:frac:small-contribution}
\end{equation}
\end{proposition}

\begin{proof}
Corollary~\ref{cor:frac:good-horizon},
Lemma~\ref{lem:frac:small-pathwise},
inclusion~\eqref{eq:frac:two-deviations}, and
Lemma~\ref{lem:frac:concentration} give
\[
    \P{T_{x}^{\DynoSPLIT}>t,\ x\le\ell_t}
    \le o\bigl(\Fbar(\delta t)\bigr)
              +2\exp\left\{-\frac{m_t t}{8\ell_t}\right\}.
\]
Since \(m_t=\varepsilon(\hatA_t)/2\), condition
\eqref{eq:frac:condition-small-exponent} makes the exponential term decay
faster than every power of \(t\).  A regularly varying tail decays more
slowly than \(t^{-p}\) for some fixed \(p>\alpha\).  The exponential term is
therefore \(o(\Fbar(\delta t))\), proving \eqref{eq:frac:small-contribution}.
\end{proof}

This proves the first bound in the roadmap.  Its main restriction is that
\(\ell_t\) must be small enough for the available margin
\(\varepsilon(\hatA_t)\) to absorb all capped work.  The large-job
argument creates the opposite pressure: \(\ell_t\) must be large enough to
make a collision with another large job negligible.

\subsection{Large jobs: collisions are rare}
\label{sec:frac:large}

We now prove the large-job bound~\eqref{eq:frac:roadmap-large}.  Recall the
idea from the roadmap: A large tagged job first waits for the current
anchor to complete, then becomes the anchor itself and keeps that role
until it finishes.  It can be late only if the wait is long, which requires
a job larger than \(\sqrt x\) to be present on arrival,\endnote{Here \(\sqrt x\) is not crucial and can be any
other function in \(o(x)\) as \(x\to\infty\).} or if a larger job
arrives during the wait and takes the anchor role first.
Lemma~\ref{lem:frac:collision} makes this precise for a fixed tagged size
and bounds both collision events using
Lemma~\ref{lem:frac:anchor-occupancy}.  Proposition~\ref{prop:frac:large}
averages the bound over the tagged size.  This is where the requirement
that pushes \(\ell_t\) upward appears.

Fix a tagged size \(x\in(\ell_t,h_t]\).  We first record how the tagged
job proceeds when nothing interferes with it.  On arrival, it joins the
non-anchor side.  If nothing interferes, this means every job present has original size at most \(\sqrt x\), and the current anchor has remaining size at most \(\sqrt x\) and is served at rate at least \(\delta/2\); the anchor then completes within
\(2\sqrt x/\delta\) time units.  If moreover no job larger than \(x\)
arrives before then, the tagged job is the largest job present when the
anchor completes, so it becomes the anchor and is served at rate at least
\(\gamma(x)\ge\gamma(\ell_t)\) until it finishes, by monotonicity of
\(\varepsilon\).  Its response time is then at most
\(2\sqrt x/\delta+x/\gamma(\ell_t) \leq 2\sqrt t/\delta+x/\gamma(\ell_t)\).  Every job in the large range
finishes within time \(t\) along this path provided
\begin{equation}
    \ell_t<h_t\le\gamma(\ell_t)\Bigl(t-\frac{2\sqrt t}{\delta}\Bigr)
    \quad\text{for all sufficiently large \(t\)}.
        \label{eq:frac:condition-completion}
\end{equation}
The term \(2\sqrt t/\delta\) is the \(o(t)\) slack of the roadmap, reserved for the wait for the current anchor.

Besides \eqref{eq:frac:condition-completion}, the other condition we require is
\begin{equation}
    \ell_t\longrightarrow\infty.
    \label{eq:frac:condition-large-threshold}
\end{equation}
With these two requirements in place, the two ways in which the path above
can fail give the following bound.

\begin{lemma}[Conditional collision bound]
\label{lem:frac:collision}
Suppose \eqref{eq:frac:condition-completion} and
\eqref{eq:frac:condition-large-threshold} hold.
For all sufficiently large \(t\) and every
\(\ell_t<x\le h_t\),
\begin{equation}
    \P{T_x^{\DynoSPLIT}>t}
    \le \frac{4r^{>\sqrt x}}{\delta}
       +\frac{2\lambda\sqrt x}{\delta}\Fbar(x).
    \label{eq:frac:collision-bound}
\end{equation}
\end{lemma}

% That is, a large tagged job is late only if a job larger than \(\sqrt x\)
% is present when it arrives, or a job larger than \(x\) arrives while it
% waits for the anchor role; the two terms bound the probabilities of these
% two events.

\begin{proof}
Let \(C_0(x)\) be the event that a job with original size larger than
\(\sqrt x\) is present at time \(0-\), and let \(C_1(x)\) be the event that
a job larger than \(x\) arrives during \((0,2\sqrt x/\delta]\).  Since
\(x>\ell_t\to\infty\), we may assume \(x>\sqrt x\).  On
\(C_0(x)^c\cap C_1(x)^c\), no job larger than \(x\) is present at time
\(0-\) or arrives before the current anchor completes, and ties with jobs
of size exactly \(x\) are broken in favor of the tagged job by arrival
order.  The path described before \eqref{eq:frac:condition-completion} therefore applies: the tagged
job becomes the anchor within \(2\sqrt x/\delta\) time units and completes
within a further \(x/\gamma(\ell_t)\).  Since \(x\le h_t\le t\) and by
condition~\eqref{eq:frac:condition-completion},
\[
    \frac{2\sqrt x}{\delta}+\frac{x}{\gamma(\ell_t)} \le \frac{2\sqrt h_t}{\delta}+\frac{h_t}{\gamma(\ell_t)}
        \le\frac{2\sqrt t}{\delta}+\frac{h_t}{\gamma(\ell_t)}
    \le t,
\]
which shows
\[
    \{T_x^{\DynoSPLIT}>t\}\subseteq C_0(x)\cup C_1(x).
\]

By PASTA and Lemma~\ref{lem:frac:anchor-occupancy} at threshold
\(a=\sqrt x\), the probability \(\P{C_0(x)}\) is at most
\(4r^{>\sqrt x}/\delta\).  Arrivals larger than \(x\) form a Poisson
process with rate \(\lambda\Fbar(x)\), so \(\P{C_1(x)}\) is at most \((2\sqrt x/\delta)\lambda\Fbar(x)\).  A union bound proves~\eqref{eq:frac:collision-bound}.
\end{proof}

Both terms of~\eqref{eq:frac:collision-bound} are at most a constant
multiple of \(r^{>\sqrt x}\) (to be proved in the next proposition). Hence we need the following requirement to make $\P{T_{x}^{\DynoSPLIT}>t,\ \ell_t<x\le h_t}$ negligible:
% # here no need to explain too much. Averaging over jobs is very normal and nothing non-trivial is here.
% , so
% averaging over the tagged size costs at most the factor \(\Fbar(\ell_t)\).
% Making the resulting product negligible relative to the super-server lower
% bound is the requirement that pushes \(\ell_t\) upward:
\begin{equation}
    r^{>\sqrt{\ell_t}}\,\Fbar(\ell_t)
    =o\bigl(\Fbar(\delta t)\bigr).
    \label{eq:frac:condition-collision}
\end{equation}

\begin{proposition}[Large-job contribution]
\label{prop:frac:large}
Under
\eqref{eq:frac:condition-completion},
\eqref{eq:frac:condition-large-threshold}, and
\eqref{eq:frac:condition-collision},
\begin{equation}
    \P{T_{x}^{\DynoSPLIT}>t,\ \ell_t<x\le h_t}
    =o\bigl(\Fbar(\delta t)\bigr).
    \label{eq:frac:large-contribution}
\end{equation}
\end{proposition}

\begin{proof}
Fix \(x>\ell_t\).  Since \(x\ge\sqrt x\), we have
\(\Fbar(x)\le\Fbar(\sqrt x)\), and
\[
    \lambda\sqrt x\,\Fbar(\sqrt x)
    \le\lambda\E{S\one\{S>\sqrt x\}}
    =r^{>\sqrt x}.
\]
Hence the second term of~\eqref{eq:frac:collision-bound} is at most
\(2r^{>\sqrt x}/\delta\), and because \(r^{>y}\) is nonincreasing in \(y\),
\[
    \P{T_x^{\DynoSPLIT}>t}
    \le\frac{6r^{>\sqrt x}}{\delta}
    \le\frac{6r^{>\sqrt{\ell_t}}}{\delta}
        \qquad\text{for all \(\ell_t<x\le h_t\)}.
\]
Conditioning on the tagged size and integrating over \(\ell_t<x\le h_t\)
multiplies this uniform bound by at most \(\Fbar(\ell_t)\), so
\[
    \P{T_{x}^{\DynoSPLIT}>t,\ \ell_t<x\le h_t}
        \le\frac{6}{\delta}\,r^{>\sqrt{\ell_t}}\,\Fbar(\ell_t),
\]
which is \(o(\Fbar(\delta t))\) by~\eqref{eq:frac:condition-collision}.
\end{proof}

% The first two contributions in the roadmap are now negligible.  The
% remaining contribution depends only on the upper threshold \(h_t\).

\subsection{Huge jobs: matching the tail boundary}
\label{sec:frac:huge}

We now prove the huge-job bound~\eqref{eq:frac:roadmap-huge}.  As in the
roadmap, huge jobs are controlled directly by their size tail, and the
proof does not need a completion argument:
\begin{equation}
    \P{T_{x}^{\DynoSPLIT}>t,\ x>h_t}
    \le \P{x>h_t}=\Fbar(h_t).
    \label{eq:frac:huge-trivial-bound}
\end{equation}
This bound determines the only additional requirement on the upper
threshold:
\begin{equation}
    \frac{h_t}{\delta t}\longrightarrow1.
    \label{eq:frac:condition-huge-boundary}
\end{equation}
Note that the completion condition~\eqref{eq:frac:condition-completion} and
\(\gamma(\ell_t)<\delta\) already give \(h_t<\delta t\).
% , so
% \eqref{eq:frac:condition-huge-boundary} only asks for the matching
% asymptotic lower bound.
In addition, when condition \eqref{eq:frac:condition-huge-boundary} holds, for every fixed \(\eta\in(0,1)\), we have
\[
    (1-\eta)\delta t\le h_t\le\delta t
\]
for all sufficiently large \(t\).  Monotonicity and regular variation give, as $t\to\infty$:
\[
    1
    \le \frac{\Fbar(h_t)}{\Fbar(\delta t)}
    \le \frac{\Fbar((1-\eta)\delta t)}{\Fbar(\delta t)}
    \longrightarrow (1-\eta)^{-\alpha}.
\]
Letting \(\eta\downarrow0\) yields
\begin{equation}
    \Fbar(h_t)\sim\Fbar(\delta t).
        \label{eq:frac:huge-tail}
\end{equation}

\subsection{Choosing the thresholds and completing the proof}
\label{sec:frac:combining}

We now fix the capacity margin and prove Theorem~\ref{thm:frac:optimality}
in the following explicit form.

\begin{theorem}[Strong tail optimality of \DynoSPLIT, explicit form]
\label{thm:frac:explicit}
Suppose \(n\ge2\), \((n-1)/n\le\rho<1\), and the job-size distribution has a
regularly varying tail with tail index \(\alpha>1\).  Let \DynoSPLIT use
the capacity margin
\begin{equation}
    \varepsilon(z):=\frac{\delta}{2+\log(1+z)}.
    \label{eq:frac:explicit-margin}
\end{equation}
Then
\[
    \lim_{t\to\infty}
    \frac{\P{T^{\DynoSPLIT}>t}}{\P{S>\delta t}}
    =1.
\]
\end{theorem}

\begin{proof}
The three bounds of the roadmap were proved under requirements on the
capacity margin and on the thresholds \(\hatA_t\), \(\ell_t\), and
\(h_t\), which we first collect.
\begin{enumerate}\renewcommand{\theenumi}{\roman{enumi}}\renewcommand{\labelenumi}{\textup{(\theenumi)}}
    \item The capacity margin: \(0<\varepsilon(z)\le\delta/2\),
          \(\varepsilon\) is nonincreasing, and \(\varepsilon(z)\to0\) as
          \(z\to\infty\); these are \eqref{eq:frac:margin-basic} and
          \eqref{eq:frac:condition-margin-shape}.
    \item The anchor threshold: \(\hatA_t\to\infty\) and
          \(\ell_t\ll\hatA_t\), which is \eqref{eq:frac:condition-cap}, and
          \(r^{>\hatA_t}+t\Fbar(\hatA_t)=o(\Fbar(\delta t))\), which is
          \eqref{eq:frac:condition-anchor}.
    \item The small/large cutoff: \(\ell_t\to\infty\), which is
          \eqref{eq:frac:condition-large-threshold};
          \(\varepsilon(\hatA_t)t/(\ell_t\log t)\to\infty\), which is
          \eqref{eq:frac:condition-small-exponent};
          \(\E{(S\wedge\ell_t)^2}/\ell_t=o(\varepsilon(\hatA_t))\), which is
          \eqref{eq:frac:condition-small-moment}; and
          \(r^{>\sqrt{\ell_t}}\,\Fbar(\ell_t)=o(\Fbar(\delta t))\), which is
          \eqref{eq:frac:condition-collision}.
    \item The large/huge cutoff:
          \(\ell_t<h_t\le\gamma(\ell_t)(t-2\sqrt t/\delta)\), which is
          \eqref{eq:frac:condition-completion}, and \(h_t/(\delta t)\to1\),
          which is \eqref{eq:frac:condition-huge-boundary}.
\end{enumerate}

It is straightforward that the capacity margin \eqref{eq:frac:explicit-margin}  satisfies (i).  For the thresholds we take
\begin{equation}
    \hatA_t:=\exp\{(\log t)^2\},
    \qquad
    \ell_t:=\frac{t}{(\log t)^6},
    \qquad
    h_t:=\gamma(\ell_t)\Bigl(t-\frac{2\sqrt t}{\delta}\Bigr)=(\delta - \varepsilon(z))\Bigl(t-\frac{2\sqrt t}{\delta}\Bigr).
    \label{eq:frac:explicit-thresholds}
\end{equation}
Lemma~\ref{lem:frac:rv-verification} in
Appendix~\ref{app:frac:rv} verifies that these choices satisfy (ii)--(iv).

With all requirements in place, split the tagged job by size as in
\eqref{eq:frac:three-way-split}:
\begin{align*}
    \P{T^{\DynoSPLIT}>t}
    &\le \P{T_{x}^{\DynoSPLIT}>t,\ x\le\ell_t}\\
    &\quad+\P{T_{x}^{\DynoSPLIT}>t,\ \ell_t<x\le h_t}
       +\P{x>h_t}.
\end{align*}
Propositions~\ref{prop:frac:small} and \ref{prop:frac:large} make the
first two terms \(o(\Fbar(\delta t))\), and \eqref{eq:frac:huge-tail}
makes the last term \((1+o(1))\Fbar(\delta t)\).  Hence we have that
\[
    \limsup_{t\to\infty}
    \frac{\P{T^{\DynoSPLIT}>t}}{\P{S>\delta t}}
        \le1.
\]
This matches the super-server lower bound~\eqref{eq:lower bound super server}, and completes the proof of strong tail optimality.\end{proof}

% !TEX root = ../main.tex
\section{Evaluation}
\label{sec:evaluation}

We evaluate \DynoSPLIT by simulation in two high-load settings and compare it with SRPT-$n$ and \SPLIT. We plot the \emph{normalized tail} $\P{T>t}/\P{S>n(1-\rho)t}$, namely the response time tail divided by the super-server lower bound of Section~\ref{sec:theory:lowerbound}. By~\eqref{eq:lower bound super server}, the normalized tail of every policy is asymptotically at least one. A normalized tail converging to one demonstrates strong tail optimality. 

Figure~\ref{fig:eval} shows the normalized tails of SRPT-$n$ in blue, \SPLIT in red, and \DynoSPLIT in orange, the latter with the capacity margin of Theorem~\ref{thm:frac:explicit}. In both settings, 5(a) and 5(b), the normalized tail of \DynoSPLIT rises to a peak and then decreases toward one, while those of the other two policies do not.\endnote{The normalized tail of \DynoSPLIT does not reach one in Figure~\ref{fig:eval} because Theorem~\ref{thm:frac:optimality} concerns the limit $t\to\infty$. With the $10^{10}$ simulated jobs, which take about four days to simulate, the response time tail can be estimated reliably only up to $t\approx10^4$. Simulating more jobs would extend the curves to larger $t$, and we expect the decrease toward one to continue.} The normalized tail of SRPT-$n$ levels off above one, which is consistent with the weak tail optimality of SRPT-$n$ \citep{split}. The normalized tail of \SPLIT grows without bound: As described in Section~\ref{sec:intro:highload}, whenever a huge job holds the dedicated server, the other $n-1$ servers face a load they cannot carry, and every job arriving in the meantime suffers a response time comparable to the size of the huge job. 

\begin{figure}[t]
\FIGURE
{% Subcaptions and graphics
\subcaptionbox{$n=2$, $\rho=0.55$, $\delta=0.9$\label{subfig:eval:n2}}
{\includegraphics[width=0.48\textwidth]{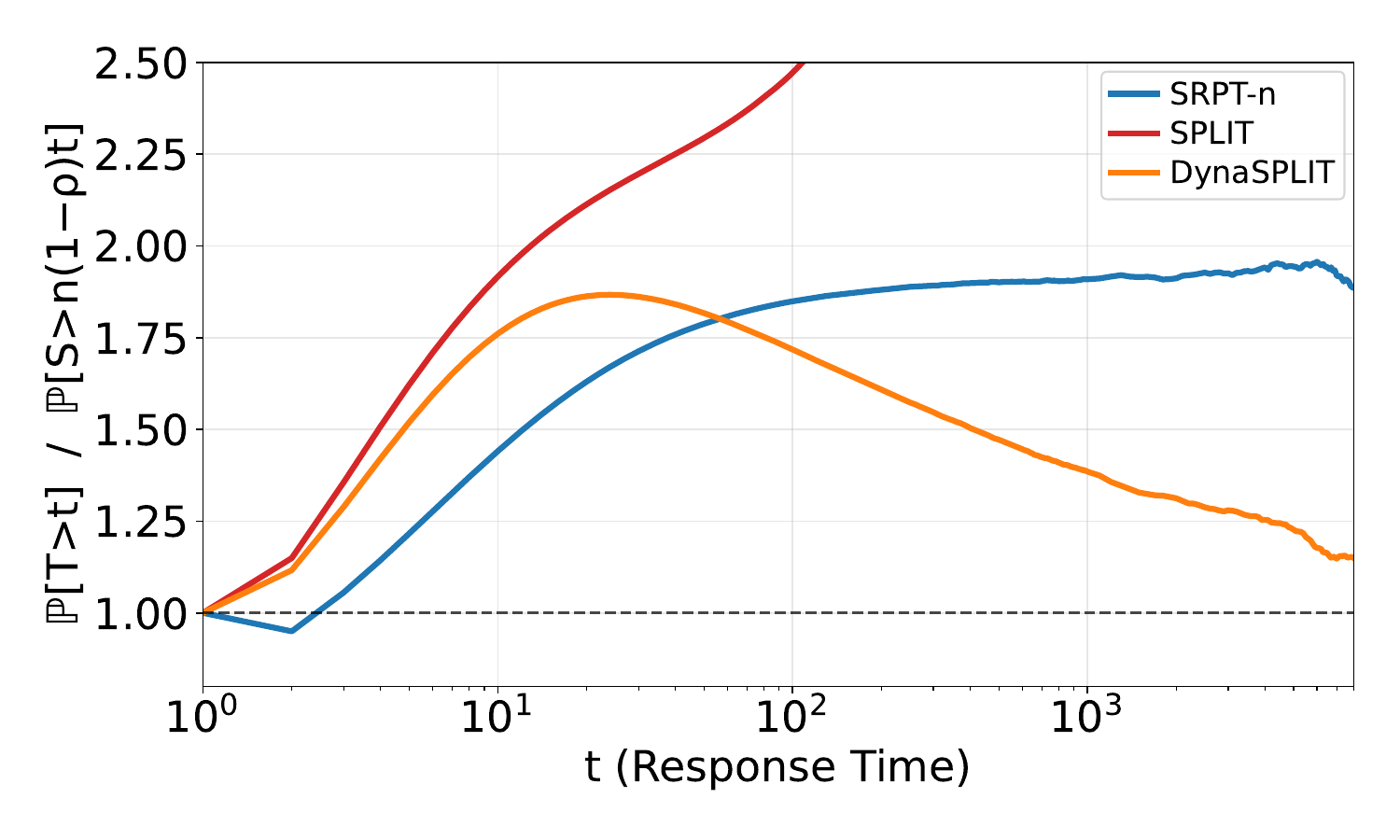}}
\hfill\subcaptionbox{$n=10$, $\rho=0.92$, $\delta=0.8$\label{subfig:eval:n10}}
{\includegraphics[width=0.48\textwidth]{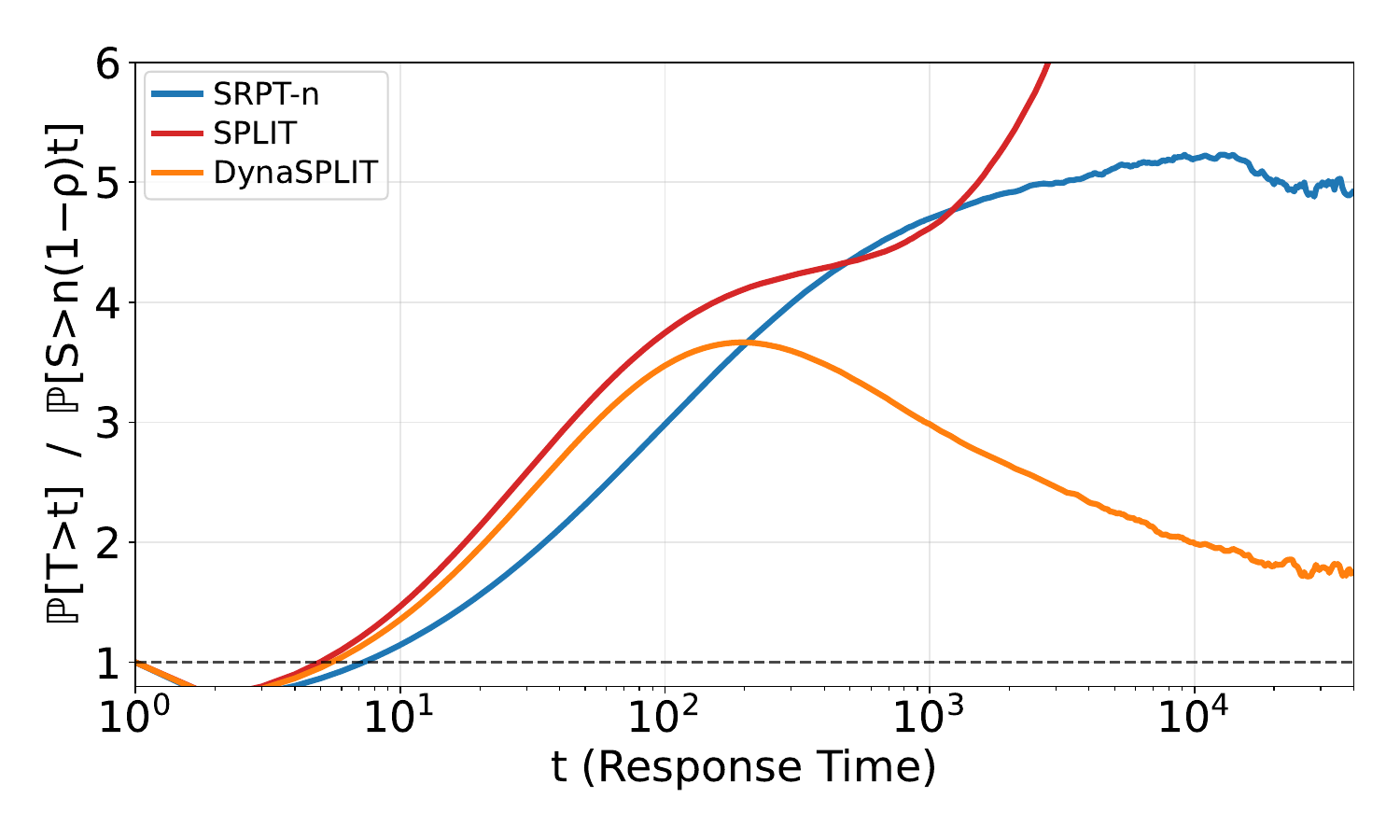}}
}
{Normalized response time tail $\P{T>t}/\P{S>n(1-\rho)t}$ from $10^{10}$ simulated jobs; Pareto job sizes with $\P{S>t}=t^{-1.8}$ for $t\ge1$.\label{fig:eval}}
{}
% Alt text: Two log-scale plots of the normalized tail against the response time. In both, the SPLIT curve grows without bound, the SRPT-n curve rises and levels off above one, and the DynaSPLIT curve rises to a peak and then decreases toward the dashed line at one; in the plot for two servers it ends close to the dashed line.
\end{figure}

% !TEX root = ../main.tex
% Related work is placed after the technical sections: the heavy-tailed
% background is in the introduction, and this section covers the
% light-tailed world, which is only loosely related to this paper.
\section{Related Work in the Light-tailed world}
\label{sec:related-work}

The heavy-tailed world, to which this paper belongs, was reviewed in Section~\ref{sec:introduction}. It has been shown that tail-optimal scheduling is very different for heavy-tailed and light-tailed job sizes \citep{wierman2012tail}. This section reviews the light-tailed world, which Section~\ref{sec:introduction} set aside.

\paragraph{Single-server systems.}
Under light-tailed job sizes, the response time tail decays exponentially under FCFS \citep{abate1995exponential}, and FCFS is weakly tail optimal, that is, its decay rate is the best possible \citep{boxma2007tails}. Size-based policies fare worse: Processor Sharing has a slower decay rate \citep{mandjes2006large}, and SRPT has the worst decay rate of all, because a large job can be starved by a stream of smaller arrivals \citep{nuyens2006large}. FCFS is not strongly tail optimal, however. Nudge \citep{grosof2021nudge} improves on FCFS by letting a small job overtake the large job immediately ahead of it, at most once, and this idea was refined in \citet{vanhoudt2022nudge,charlet2024nudgem}. Boost \citep{yu2024strongly} is the first strongly tail optimal policy, and Harlev et al.~\citeyearpar{harlev2025gittins} extend strong tail optimality to unknown job sizes with a Gittins-based policy.

\paragraph{Multi-server systems.}
The only result on tail-optimal scheduling for multi-server systems with light-tailed job sizes is due to Yu et al.~\citeyearpar{yu2025tale}: Boost is strongly tail optimal in the M/G/$n$ in the heavy-traffic limit $\rho\to1$, whereas its performance degrades away from heavy traffic. The same paper proposes CombinedBoost, a heuristic that gives large jobs some priority in order to keep all servers busy, and shows empirically that it performs well at every load, but without a proof. 

% !TEX root = ../main.tex
\section{Conclusion}
\label{sec:conclusion}

This paper gives the first strongly tail optimal scheduling policy for the M/G/$n$ with heavy-tailed job sizes in the high-load regime. It thereby closes the open problem posed by Li et al.~\citeyearpar{split}: Together with \SPLIT, strongly tail optimal policies now exist over the entire stability region. Our central finding is that the high-load regime changes what optimality demands. The best possible tail is that of a single server of speed $n$, and matching it leaves no room for a fixed reserve: A huge job must receive essentially all of the spare capacity, while the other jobs need a positive share of that same capacity to keep up with their arrivals. \DynoSPLIT resolves this dilemma by making the division of the spare capacity depend on the size of the largest job, with a margin that is positive for every finite size and vanishes as the size grows.

\begingroup \parindent 0pt \parskip 0.0ex \def\enotesize{\normalsize} \theendnotes \endgroup

% Appendix here
% Options are (1) APPENDIX (with or without general title) or
%             (2) APPENDICES (if it has more than one unrelated sections)
\begin{APPENDICES}
% !TEX root = ../main.tex
\section{Verification of the Threshold Conditions}
\label{app:frac:rv}

This appendix proves Lemma~\ref{lem:frac:rv-verification}: the capacity
margin~\eqref{eq:frac:explicit-margin} and the
thresholds~\eqref{eq:frac:explicit-thresholds} chosen in
Section~\ref{sec:frac:combining} satisfy every condition listed in the
proof of Theorem~\ref{thm:frac:explicit}.  The verification uses two
standard facts about a regularly varying tail with index \(\alpha>1\)
\citep{bgt}.  The first is Potter's bound: for every \(\eta>0\), there are
constants \(C_\eta\ge1\) and \(x_\eta\) such that, for all
\(y\ge x\ge x_\eta\),
\begin{equation}
    C_\eta^{-1}\left(\frac{y}{x}\right)^{-\alpha-\eta}
    \le \frac{\Fbar(y)}{\Fbar(x)}
    \le C_\eta\left(\frac{y}{x}\right)^{-\alpha+\eta}.
    \label{eq:frac:potter}
\end{equation}
Taking \(x=x_\eta\) in the upper bound shows in particular that
\(\Fbar(y)\le Cy^{-\alpha+\eta}\) for all large \(y\), with a constant
\(C\) depending on \(\eta\).  The second is Karamata's theorem, which
bounds the rate of work arriving in jobs larger than \(y\).

\begin{lemma}[Tail work]
\label{lem:frac:rv-bounds}
There is a constant \(C<\infty\) such that, for all sufficiently large
\(y\),
\begin{equation}
    r^{>y}\le Cy\Fbar(y).
    \label{eq:frac:rv-bounds}
\end{equation}
\end{lemma}

\begin{proof}
The tail-integral identity gives
\(\E{S\one\{S>y\}}=y\Fbar(y)+\int_y^\infty\Fbar(u)\,du\), and Karamata's
theorem gives \(\int_y^\infty\Fbar(u)\,du\sim y\Fbar(y)/(\alpha-1)\)
because \(\alpha>1\).  Multiplying by \(\lambda\)
proves~\eqref{eq:frac:rv-bounds}.
\end{proof}

Throughout the rest of this appendix, \(C\) denotes a finite constant
whose value may change from line to line.

\begin{lemma}[Verification of the threshold conditions]
\label{lem:frac:rv-verification}
Let \(\varepsilon\) be given by \eqref{eq:frac:explicit-margin} and
\(\hatA_t\), \(\ell_t\), \(h_t\) by \eqref{eq:frac:explicit-thresholds}.
Then conditions
\eqref{eq:frac:condition-cap},
\eqref{eq:frac:condition-anchor},
\eqref{eq:frac:condition-large-threshold},
\eqref{eq:frac:condition-small-exponent},
\eqref{eq:frac:condition-small-moment},
\eqref{eq:frac:condition-collision},
\eqref{eq:frac:condition-completion}, and
\eqref{eq:frac:condition-huge-boundary} hold.
\end{lemma}

\begin{proof}
Throughout, choose \(\eta\in(0,\alpha-1)\) and write
\(\beta:=\alpha-\eta>1\).  The margin at the anchor threshold is
\begin{equation}
    \varepsilon(\hatA_t)
    =\frac{\delta}{2+\log\bigl(1+e^{(\log t)^2}\bigr)}
    =\Theta\bigl((\log t)^{-2}\bigr).
    \label{eq:frac:explicit-cap-margin}
\end{equation}

\emph{Conditions \eqref{eq:frac:condition-cap},
\eqref{eq:frac:condition-large-threshold}, and
\eqref{eq:frac:condition-small-exponent}.}
Clearly \(\hatA_t\to\infty\) and \(\ell_t\to\infty\), and
\(\ell_t/\hatA_t\to0\) because \(\hatA_t=t^{\log t}\) grows faster than
every fixed power of \(t\).  By~\eqref{eq:frac:explicit-cap-margin},
\[
    \frac{\varepsilon(\hatA_t)t}{\ell_t\log t}
        =\Theta\bigl((\log t)^3\bigr)
    \longrightarrow\infty.
\]

\emph{Conditions \eqref{eq:frac:condition-completion} and
\eqref{eq:frac:condition-huge-boundary}.}
Because \(\ell_t/t\to0\) and \(\gamma(\ell_t)\to\delta>0\), we have
\(\ell_t<h_t\) eventually, and \(h_t\) satisfies the upper inequality
in~\eqref{eq:frac:condition-completion} with equality.  Moreover,
\[
    \frac{h_t}{\delta t}
    =\frac{\gamma(\ell_t)}{\delta}\Bigl(1-\frac{2}{\delta\sqrt t}\Bigr)
    \longrightarrow1.
\]

\emph{Condition \eqref{eq:frac:condition-anchor}.}
Lemma~\ref{lem:frac:rv-bounds} gives \(r^{>\hatA_t}\le C\hatA_t\Fbar(\hatA_t)\),
and the upper bound in~\eqref{eq:frac:potter} with \(x=\delta t\) and
\(y=\hatA_t\) gives \(\Fbar(\hatA_t)/\Fbar(\delta t)\le C(\hatA_t/t)^{-\beta}\).
Hence we have that
\begin{align*}
    \frac{r^{>\hatA_t}+t\Fbar(\hatA_t)}{\Fbar(\delta t)}
    &\le C(\hatA_t+t)\left(\frac{\hatA_t}{t}\right)^{-\beta}\\
    &=C\left(
        t^\beta\hatA_t^{1-\beta}
        +t^{\beta+1}\hatA_t^{-\beta}
      \right).
\end{align*}
Because \(\hatA_t=t^{\log t}\) grows faster than every fixed power of
\(t\), both terms converge to zero.

\emph{Condition \eqref{eq:frac:condition-small-moment}.}
Without loss of generality take \(\beta<2\), which is possible because
\(\eta\) may be increased toward \(\alpha-1\).  The tail-integral identity
for the capped second moment gives
\[
    \E{(S\wedge\ell_t)^2}
    =2\int_0^{\ell_t}x\Fbar(x)\,dx.
\]
Using \(\Fbar(x)\le Cx^{-\beta}\) for \(x\ge x_\eta\) and
\(\Fbar\le1\) below, the integral is at most
\(C+C\ell_t^{2-\beta}\), and therefore
\[
    \frac{\E{(S\wedge\ell_t)^2}}{\ell_t}
    \le C\ell_t^{1-\beta}+\frac{C}{\ell_t}
    =O\bigl(\ell_t^{1-\beta}\bigr).
\]
The right side decays polynomially in \(t\), because
\(\ell_t=t/(\log t)^6\) and \(\beta>1\), whereas
\eqref{eq:frac:explicit-cap-margin} says that \(\varepsilon(\hatA_t)\)
decays only logarithmically.

\emph{Condition \eqref{eq:frac:condition-collision}.}
Lemma~\ref{lem:frac:rv-bounds} gives
\(r^{>\sqrt{\ell_t}}\le C\sqrt{\ell_t}\,\Fbar(\sqrt{\ell_t})\).
Since \(\ell_t<\delta t\) for all large \(t\), the lower bound
in~\eqref{eq:frac:potter} with \(x=\ell_t\) and \(y=\delta t\) gives
\[
    \frac{\Fbar(\ell_t)}{\Fbar(\delta t)}
    \le C\left(\frac{\delta t}{\ell_t}\right)^{\alpha+\eta}
    =C\delta^{\alpha+\eta}(\log t)^{6(\alpha+\eta)}.
\]
Combining this with the polynomial bound
\(\Fbar(\sqrt{\ell_t})\le C\ell_t^{-\beta/2}\) yields
\[
    \frac{r^{>\sqrt{\ell_t}}\,\Fbar(\ell_t)}{\Fbar(\delta t)}
    \le C\ell_t^{(1-\beta)/2}(\log t)^{6(\alpha+\eta)}
    \le Ct^{(1-\beta)/2}(\log t)^{3(\beta-1)+6(\alpha+\eta)}
    \longrightarrow0,
\]
because \(\beta>1\).  The only essential fact in this calculation is
\(\alpha>1\): it makes the negative power of \(t\) dominate the logarithmic
factor.
\end{proof}

% Alter-SPLIT appendix removed from the paper (2026-10-05); the file is kept, not compiled.
% \input{sections/B-alter-split}
\end{APPENDICES}

% Acknowledgments here
\ACKNOWLEDGMENT{This work was supported by NSF-CIF-2403194, NSF-III-2322973, and NSF-CMMI-2307008.}

% References here
\bibliographystyle{INFORMS-STSY-Template-26-02-2026/informs2014}
\bibliography{bib}

%%%%%%%%%%%%%%%%%
\end{document}